\def\version{\today}	        	%
\documentclass[reqno,11pt]{amsart} 
\usepackage[utf8]{inputenc}
\usepackage{amsmath} 
\usepackage[mathscr]{eucal}
\usepackage{amssymb}
\usepackage{srcltx} 
\usepackage{dsfont}
\usepackage{hyperref}
\usepackage{color}
\usepackage{enumerate}
\usepackage{tikz,pgfplots}
\usetikzlibrary{arrows, positioning}
\usepackage{comment}
\usepackage{lscape}
\usepackage{graphicx}
\usepackage{tabularx}
\usepackage{caption}
\usepackage{diagbox}
\usepackage{subcaption}
\pgfplotsset{compat=1.11}
\usepgfplotslibrary{fillbetween}
\usetikzlibrary{intersections}
\pgfdeclarelayer{bg}
\pgfdeclarelayer{ft}
\pgfsetlayers{bg,main,ft}

\numberwithin{equation}{section}

\newfam\Bbbfam 
\font\tenBbb=msbm10 
\font\sevenBbb=msbm7 
\font\fiveBbb=msbm5 
\textfont\Bbbfam=\tenBbb 
\scriptfont\Bbbfam=\sevenBbb 
\scriptscriptfont\Bbbfam=\fiveBbb

\def\2{\mathbf 2}

\newcommand{\R}     {\mathbb{R}} 
\newcommand{\Z}     {\mathbb{Z}} 
\newcommand{\N}     {\mathbb{N}} 
 
\newcommand{\D}     {\mathbb{D}}

\def\1{{\mathchoice {1\mskip-4mu\mathrm l}      
		{1\mskip-4mu\mathrm l} 
		{1\mskip-4.5mu\mathrm l} {1\mskip-5mu\mathrm l}}} 
 
\def\comment#1{} 
\newtheoremstyle{thm}{2ex}{2ex}{\itshape\rmfamily}{} 
{\bfseries\rmfamily}{}{1.7ex}{} 

\newtheoremstyle{rem}{1.3ex}{1.3ex}{\rmfamily}{} 
{\itshape\rmfamily}{}{1.5ex}{}

\newtheorem{theorem}{Theorem}[section] 
 
\newtheorem{prop}[theorem] {Proposition}

\theoremstyle{definition}
\newtheorem{defn}[theorem] {Definition} 
\newtheorem{example}[theorem] {Example}

\newtheorem{remark}[theorem]{Remark} 
\newtheorem{assumption}{Assumption}

\newcommand{\Mcal}   {{\mathcal M }}

\newcommand{\Ucal}   {{\mathcal U }}

\newcommand{\Xcal}   {{\mathcal X }}

\definecolor{Red}{rgb}{1,0,0}

\begin{document} 
	
	\title[The canonical equation of adaptive dynamics with power-law mutation rates]
	{The canonical equation of adaptive dynamics in individual-based models with power-law mutation rates}
	\author[Tobias Paul]{}
	\maketitle
	\thispagestyle{empty}
	\vspace{-0.5cm}
	
	\centerline{Tobias Paul{\footnote{HU Berlin, Rudower Chaussee 25, 12489 Berlin, {\tt tobias.paul.1@hu-berlin.de}}}}
	\renewcommand{\thefootnote}{}
	\vspace{0.5cm}
	
	\bigskip
	
	\centerline{\small(\version)} 
	\vspace{.5cm} 
	
	\begin{quote} 
		{\small {\bf Abstract:}} In this paper, we consider an individual-based model with power-law mutation probability. In this setting, we use the large population limit with a subsequent ``small mutations'' limit to derive the canonical equation of adaptive dynamics. For a one-dimensional trait space this corresponds to well established results and we can formulate a criterion for evolutionary branching in the spirit of Champagnat and Méléard (2011). However, for more complex models higher dimensional trait spaces are required model various aspects of coexisting individuals without simplifying potential trade-offs. In higher dimensional trait spaces, we find that the speed at which the solution of the canonical equation moves through space is reduced due to mutations being restricted to the underlying grid on the trait space. However, as opposed to the canonical equation with rare mutations, we can explicitly calculate the path which the dominant trait will take without having to solve the equation itself.
	\end{quote}
	
	

	\setcounter{tocdepth}{3}
	

	\setcounter{section}{0}
	\begin{comment}{
	This is not visible.}
	\end{comment}
	
	\section{Introduction}
	\label{Sec: Intro }
	
	When considering stochastic population models with mutation, we want to understand which mutations are successful and how the sequence of successful mutations behaves over time. First ideas for a single equation describing this sequence were heuristically given by \cite{MGM+96} and \cite{DL96}, discussed in the setting of Markov processes in \cite{CFB01} and later made rigorous by \cite{CM11}. They consider the ``rare'' mutation regime, where the probability for a mutation at birth $u_K$ satisfies \begin{align}
	e^{-VK}\ll u_K\ll\frac{1}{K\log K}\label{eq: rare mutations}
	\end{align}
	for all $V>0$, where $K$ is a scaling parameter for the population size which in this context is called \emph{carrying capacity}. Here, we write $f(K)\ll g(K)$ if $f(K)/g(K)\to 0$ as $K\to\infty$. This scaling leads, with $K\to\infty$, to the \emph{trait substitution sequence} as was shown in \cite{C06}. As the name implies, this limit describes a sequence of traits which dominate the population. In particular, at each point in time there is a unique dominating trait. When one takes the subsequent limit for the radius of mutation (i.e. a parameter determining how much a mutant trait differs from the parental trait) $\sigma\to 0$, the trait substitution sequence converges weakly to the solution of the canonical equation of adaptive dynamics \cite{CM11}. While this approach of taking limits successively is mathematically convenient, it does not allow to make any claims about the quality of approximation for a fixed set of parameters since we do not know the relation of $\sigma$, $u_K$ and $K$. This issue was resolved by \cite{BBC17} showing that a suitable scaling of $\sigma$ also depending on $K$ allows to consider the simultaneous limit and recover the same limiting equation.\\
	
	The scaling \eqref{eq: rare mutations} of the mutation probability is called ``rare'' because in the limiting process we see mutations invading the current population successively and any advantageous mutation immediately replaces the previous dominant trait before another mutation can invade. Recently, individual-based models with higher mutation probabilities, namely $u_K=K^{-\alpha}$ for some $\alpha\in(0,1)$, have been investigated. We will refer to this regime as ``power-law'' mutations. Since mutations are so much more frequent, we will require the set of possible traits which can be attained by any sequence of mutations to be finite. Then, as was shown to be true on arbitrary finite graphs by \cite{CKS21}, a logarithmic scaling of time gives rise to convergence of the exponents $\beta_x^K(t)$ of the population sizes $K^{\beta_x^K(t)}-1$ to piecewise affine functions $\beta_x(t)$ for all traits $x$ and suitable times $t\geq 0$. This result already appeared earlier in a specific model of \cite{CMT21} considering mutations and horizontal gene transfer. Similar observations regarding the connection between the exponents of a population and linear behaviour on a $\log$-time scale were already made by Durrett and Mayberry \cite{DM11} in a population genetic setting with fixed population size, but were also described in e.g.~\cite{BCS19} for estimating times to cross fitness valleys or \cite{S17} in the context of recurrent mutations. For a general overview of results for both ``rare'' and ``power law'' mutations, we refer to \cite{CMT23}. However, the canonical equation has not been discussed yet for models with power-law mutations.\\
	
	In this paper, we take the first step to obtain the canonical equation of adaptive dynamics while considering a power-law mutation probability for the underlying individual-based model beginning with the case of taking successive limits. More precisely, we will first let the mutation probability and population size tend to $0$ and $\infty$ respectively simultaneously with $K\to\infty$. Then, we let the exponent of the mutation probability $\alpha\to 1$ and only thereafter we let the radius of mutation $\sigma$ go to $0$. We are taking three limits as opposed to only two limits in the ``rare'' mutation regime in order to obtain better control over the times at which the dominating trait in the population changes. While it is desirable to obtain a corresponding result for general $\alpha\in(0,1)$, the dynamics of the large population limit are difficult to handle and hence we leave this for future work. As we will see in Section \ref{Sec: 1dim}, as long as we are in a one-dimensional trait space, we recover exactly the canonical equation that we would expect from translating the canonical equation with ``rare'' mutations into our setting.
	
	However, in order to incorporate different facets of a trait into a model, one needs higher dimensional spaces with each dimension describing one feature of a trait. Such aspects may be the reproduction rate, mortality rate,  tolerance to competitive pressure, ability to perform horizontal gene transfer or phenotypic switching \cite{BB18, BPT23}. While it is possible to consider all of these features in a one-dimensional model, this requires significant simplification of the trade-offs between the different aspects of a trait. For spaces of dimension strictly larger than $1$, the canonical equation becomes only piecewise differentiable and the speed of evolution through space is reduced. This is a contrast to the canonical equation with rare mutations where there is a closed formula independent of the dimension of the trait space. Therefore, we need to treat higher dimensions separately and will do so extensively for the special case of two dimensions in Section \ref{Sec: 2dim}. Differently from the canonical equation with rare mutations, we get an explicit description of the path of the solution of the canonical equation from the fitness function.
	
	One peculiarity of the canonical equation with power-law mutation rates is the fact that there is no term corresponding to mutational variance $\sigma^2$ which one sees frequently in other models \cite{DL96, CFB01, CH23}. This observation is rooted in the deterministic nature of the large population limit. While the rare mutation regime results in the large population limit in the polymorphic evolution sequence \cite{CM11} which is a stochastic process, the power-law mutation regime has a deterministic limit as outlined above. Therefore, there is no more randomness involved in determining the next resident trait which carries over to the small mutation limit.

	\section{The canonical equation in one dimension}\label{Sec: 1dim}
	In this section, we derive the canonical equation as the small mutation limit of an individual-based model with power-law mutations in a one-dimensional trait space. For the purpose of this section, let $I\subseteq \R$ be an interval and let $\Xcal = I\cap\delta\Z$ be the grid of size $\delta>0$ on this interval. For now we will suppress the dependency on $\delta$ in our notation but we will point to it when it becomes important. We consider a population composed of individuals with traits in the trait space $\Xcal$. The individual dynamics depend on their trait in the following way: Let $K\in\N$ and $\alpha\in(0,1)$ be fixed. \begin{itemize}
		\item At rate $b(x)\geq 0$ the individual with trait $x\in\Xcal$ gives birth to another individual which in general carries the parental trait. However, with probability $K^{-\alpha}$ the offspring mutates with equal probability to the traits $x+\delta$ or $x-\delta$. If the mutant trait is not contained in $\Xcal$, the offspring carries the parental trait.
		\item At rate $d(x)\geq 0$, the individual with trait $x\in\Xcal$ dies.
		\item At rate $c(x,y)/K\geq 0$, the individual with trait $x$ experiences competition from an individual with trait $y$ and dies as a result of this competitive event.
	\end{itemize}
	These are the standard individual dynamics in the adaptive dynamics framework where we only adjusted the mutation probability to be a power law $u_K=K^{-\alpha}$. We are interested in the dynamics of the population sizes as $K\to\infty$ and more specifically in the sequence of traits which dominate the population in a suitable sense. Following \cite{CKS21}, we denote the number of individuals carrying trait $x\in\Xcal$ at time $t\geq 0$ by $N_{x,\alpha}^K(t)$. Under a suitable rescaling of time, the population sizes behave approximately exponentially with base $K$. Therefore, we define the exponents $\beta^K_{x,\alpha}$ describing the population size of trait $x\in\Xcal$ via the relation \[
	N_{x,\alpha}^K(t \log K) = K^{\beta^K_{x,\alpha}(t)}-1\quad\Longleftrightarrow\quad \beta_{x,\alpha}^K(t) := \frac{\log(1+ N_{x,\alpha}^K(t\log K))}{\log K}. 
	\]
	Going forward, we also need to introduce the notion of fitness. To this end, denote by $\mathbf{v}\subseteq \Xcal$ a set of traits in the trait space. The traits are said to $\emph{coexist}$ if the mutation free Lotka-Volterra system \[
	\dot n_y(t) = n_y(t)\left(b(y)-d(y)-\sum_{x\in\mathbf{v}} c(y,x)n_x(t)\right),\qquad y\in\mathbf{v}
	\]
	has a unique coordinate-wise strictly positive equilibrium which we denote by $\bar{n}(\mathbf{v})$. The invasion fitness of a trait $y\in\Xcal$ in a population composed of coexisting traits $\mathbf{v}$ is defined as \begin{align}
	f(y,\mathbf{v})=b(y)-d(y)-\sum_{x\in\mathbf{v}}c(y,x)\bar{n}_x(\mathbf{v}).\label{eq: fitness}
	\end{align}
	In the cases where $\mathbf{v}=\{x\}$ is a singleton, we also write $f(y,x)$ instead of $f(y,\{x\})$. The function $f(y,\mathbf{v})$ describes the initial rate of growth of individuals with trait $y$ since it is the rate of growth of the approximating branching process. Hence, if $f(y,\mathbf{v})>0$, trait $y$ is \emph{fit} and can invade against the traits in $\mathbf{v}$ whereas if $f(y,\mathbf{v})<0$, the trait $y$ will go extinct almost surely. We call $y\in\Xcal$ a \emph{mutant trait} of the traits $\mathbf{v}$ if \[
	\mathrm{dist}(\mathbf{v},y):=\min_{x\in\mathbf{v}}\mathrm{dist}(x,y)>0,
	\]
	where $\mathrm{dist}(x,y)$ denotes the number of mutations needed to reach trait $y$ by a sequence of mutations started from $x$. This can also be seen as the length of the shortest path from $x$ to $y$ in the directed graph $(V,E)$ with vertex set $V=\Xcal$ and edges given by the possible mutations. We then arrive at a theorem detailing the limiting functions $\beta_{x,\alpha}$, which we want to sketch here.
	\begin{theorem}[Theorem 2.2 in \cite{CKS21}]\label{Thm: CKS21}
		Let $\mathbf{v}_0\subseteq\Xcal$ be a set of coexisting traits. 
		 Then under suitable further assumptions, the functions $\beta_{x,\alpha}^K(t)$ converge as $K\to\infty$ in the space of càdlàg paths on the time interval $[0,T\wedge T_0]$ for any time $T>0$ and a time $T_0$ determined later to a piecewise affine function $\beta_{x,\alpha}$ which is constructed in the following way:
		 \begin{itemize}
		 	\item We define the sequence of invasion times recursively by setting $s_0:=0$ and for $k\geq 1$ \[
		 	s_k := \inf\{t>s_{k-1}\mid \exists y_k\in\Xcal\setminus\mathbf{v}_{k-1} : \beta_{y_k,\alpha}(t)=1\},
		 	\]
		 	where $\mathbf{v}_k$ denotes the set of coexisting traits from $\mathbf{v}_{k-1}$ and the trait $y_k$ as detailed in the above definition of the time $s_{k}$. These are the traits which have strictly positive coordinates in the unique equilibrium of the mutation free Lotka-Volterra system for the traits $\mathbf{v}_{k-1}\cup\{y_k\}$. If there is no unique trait $y_k$ or no unique equilibrium, we set $T_0:= s_{k}$.
		 	\item We define \[
		 	\beta_{x,\alpha}(t)=\left(\max_{u\in\Xcal}\left(\beta_{u,\alpha}(s_{k-1})+(t-t_{u,k}\wedge t)f(u,\mathbf{v}_{k-1})-\alpha \mathrm{dist}(u,x)\right)\right)\vee 0\quad\text{for}\quad s_{k-1}\leq t \leq s_k
		 	\]
		 	where the time $t_{u,k}$ is defined as \[
		 	t_{u,k}=\begin{cases}
		 		\inf \{t\geq s_{k-1}\mid \beta_{u,\alpha} = \alpha\},&\quad\text{if }\beta_{x,\alpha}(s_{k-1})=0\\
		 		s_{k-1},&\quad\text{otherwise}.
		 	\end{cases}
		 	\]
		 	\item If $t_{u,k}=s_k$ for some $u\in\Xcal$ or $\beta_{x,\alpha}(s_k)=0$ and $\beta_{x,\alpha}(s_k-\varepsilon)>0$ for all sufficiently small $\varepsilon>0$, the construction is stopped and $T_0:=s_k$.
		 \end{itemize}
	 	
	\end{theorem}

	\begin{remark}
		Note that the functions $\beta_{x,\alpha}$ are well defined by the recursive structure for the times $s_k$. For the definition of the times $s_k$ we always require a unique trait $y$ to have an exponent reaching $1$. If there were multiple such traits, we could not control the equilibrium population size since it might depend on the starting condition of the approximating Lotka-Volterra system.
	\end{remark}
	\begin{remark}
		The functions $\beta_{x,\alpha}$ do not depend on the specific choice of the underlying mutation probabilities. That is, if there are two mutation kernels $M(x,\cdot)$ and $\tilde{M}(x,\cdot)$ determining the distribution of the trait of a mutant offspring of $x$ which for each $x\in\Xcal$ have positive mass on the same set but with possibly different values, then the limiting functions $\beta_{x,\alpha}$ will be identical. In particular, the variance of the mutation kernels is not visible in the large population limit.
	\end{remark}
	
	 We can use this convergence theorem for new terminology. We call traits $x\in\Xcal$ with limiting exponent $\beta_{x,\alpha}(t)=1$ \emph{resident} at time $t$.
 	There is a large number of technical details and assumptions related to this theorem that we have omitted for the sake of briefness. One of the main assumptions is $1/\alpha\notin\N$ since this would cause problems in the stochastic setting concerning the emergence of new traits and extinction of old traits when there is a change in the resident population. However, since we have taken the large population limit already and hence have a deterministic function, we can now take the subsequent limit $\alpha\to 1$ by setting $\alpha=1$ in the formula for $\beta_{x,\alpha}$. As we will see in a moment, this greatly helps our cause in deriving the canonical equation of adaptive dynamics in this setting. The functions $\beta_x:=\beta_{x,1}$ now take the form \begin{align}
 	\beta_x^\delta(t):=\begin{cases}
 		 [\beta_x^\delta(s_{k-1})+(t-s_{k-1})f(x,\mathbf{v}_{k-1})]\vee 0,&\quad\text{if } \mathrm{dist}(\mathbf{v}_{k-1},x)=1 \text{ or } \beta_x^\delta(s_{k-1})>0\\
 		 0 &\quad\text{otherwise.}
 	\end{cases}\label{eq: PES}
 	\end{align}
 	for $t\in[s_{k-1},s_k]$, $k\in\N$.  Here, we use the notation $\beta_x^\delta$ to highlight the dependence of the limiting function on the underlying $\delta$-grid in trait space.
    This definition allows us to determine the time until the next exponent reaches $1$, i.e.~ the time until the next mutation successfully invades the population explicitly by defining \[
 	s_k:=s_{k-1}+\min_{x\in\mathbf{v}_{k-1}}\min\left(\frac{1-\beta^\delta_{x+\delta}(s_{k-1})}{f_+(x+\delta,\mathbf{v}_{k-1})},\frac{1-\beta^\delta_{x-\delta}(s_{k-1})}{f_+(x-\delta,\mathbf{v}_{k-1})}\right).
 	\]
 	Here, $f_+:=\max(f,0)$ is the non-negative part of $f$ and we use the convention that dividing by $0$ yields $\infty$. If we had not let $\alpha\to 1$ previously, we would not be able to get this simple representation. This is due to possible changes in fitness when there is a change in the resident trait causing a change in slope. Another reason is the possibility of secondary mutants (i.e mutants of mutant traits) to have a high fitness which then could lead to a jump in the sequence of resident traits. This would require us to pay attention not only to the current resident trait, but also to the exponent size of all other traits which is difficult to handle.\\
 	
 	\begin{example}\label{Ex: minimum example}
 		We consider a minimum working example to illustrate the differences between $\beta_{x,\alpha}$ and $\beta_{x,1}$. Let $\Xcal=\{0,1,2\}$, $b(x)=1+x$, $d(x)=0$ and $c(x,y)=1$ for all pairs $(x,y)\in\Xcal^2$. Further we consider the case in which $\alpha = 3/4$. We readily calculate the invasion fitness to be \[
 		f(1,0) = 1,\quad f(2,0) = 2,\quad f(0,1) = -1, \quad f(2,1) = 1, \quad f(0,2)=-2 \quad\text{and}\quad f(1,2)=-1.
 		\]
 		Assuming the trait $0$ to be initially resident, the functions $\beta_{x,3/4}$ take the form displayed in Figure \ref{Fig: Prelimit}.
 		
 		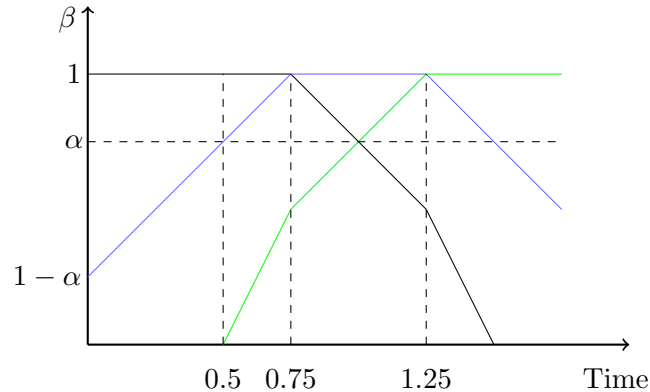
\begin{figure}[htbp]
 			\begin{tikzpicture}[scale=0.9]
 				\draw[->, thick] (0,0) -- (8,0);
 				\draw[->, thick] (0,0) -- (0,5);
 				\draw[black] (0,4) -- (3,4) -- (5,2) -- (6,0) ;
 				\draw[blue!60] (0,1) -- (3,4) -- (5,4) -- (7,2);
 				\draw[green] (2,0) -- (3,2) -- (5,4) -- (7,4) ; 
 				\draw[dashed] (0,3) -- (7,3) ; 
 				\draw[dashed] (2,0) -- (2,4) ;
 				\draw[dashed] (3,0) -- (3,4) ; 
 				\draw[dashed] (5,0) -- (5,4) ; 
 				\draw (-0.2,4) node {$1$};
 				\draw (-0.2,3) node {$\alpha$};
 				\draw (-0.6,1) node {$1-\alpha$};
 				\draw (2,-0.5) node {$0.5$};
 				\draw (3,-0.5) node {$0.75$};
 				\draw (5,-0.5) node {$1.25$};
 				\draw (7.8,-0.5) node {Time};
 				\draw (-0.3,4.8) node {$\beta$};
 				
 			\end{tikzpicture}
 			\caption{Sketch of the functions $\beta_{x,3/4}$. Black is $\beta_{0,3/4}$, blue is $\beta_{1,3/4}$ and green is $\beta_{2,3/4}$. Trait $2$ emerges when trait $1$ reaches the exponent $\alpha$ and there is a subsequent change in the fitness of trait $2$ at the time trait $1$ becomes resident. }
 			\label{Fig: Prelimit}
 		\end{figure}
 	During the residency of trait $0$, the slope of trait $2$ is twice as steep as that of trait $1$. In more extreme cases (e.g. choosing $b(x)=1+x^3$) it may hence happen that trait $1$ does not get resident at all and instead trait $2$ becomes the resident trait after trait $0$. Letting $\alpha\to 1$ prohibits both the change in slope and the possibility of skipping a trait due to a large fitness advantage of the mutant of a mutant. This is illustrated in Figure \ref{Fig: Limit}. 
 	\begin{figure}[htbp]
 		\begin{tikzpicture}[scale=0.9]
 		\draw[->, thick] (0,0) -- (11,0);
 		\draw[->, thick] (0,0) -- (0,5);
 		\draw[black] (0,4) -- (4,4) -- (8,0) ;
 		\draw[blue!60] (0,0) -- (4,4) -- (8,4) -- (10,2);
 		\draw[green] (4,0) -- (8,4) -- (10,4) ; 
 		\draw (-0.2,4) node {$1$};
 		\draw (4,-0.5) node {$1$};
 		\draw (8,-0.5) node {$2$};
 		\draw[dashed] (4,0) -- (4,4) ;
 		\draw[dashed] (8,0) -- (8,4) ;
 		\draw (10.8,-0.5) node {Time};
 		\draw (-0.3,4.8) node {$\beta$};
 	\end{tikzpicture}
 	\caption{Sketch of the functions $\beta_x$. Black is $\beta_0$, blue is $\beta_1$ and green is $\beta_2$.}
 	\label{Fig: Limit}
 	\end{figure}
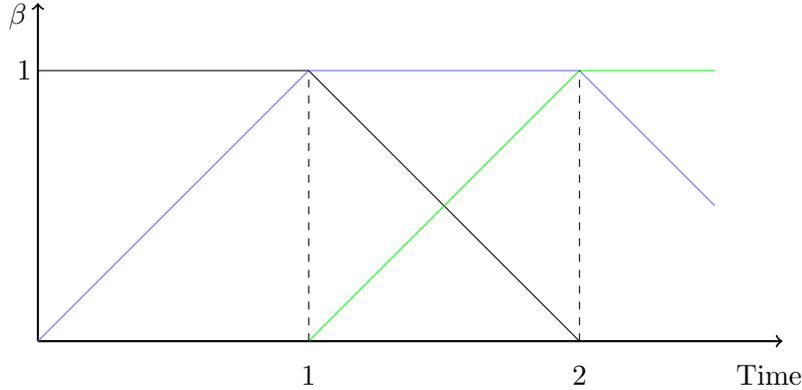
 	\end{example}
 	
 	In order to arrive at the canonical equation, we will have to take a step back and consider monomorphic populations. While the above definitions hold for any arbitrary number of coexisting traits, we will now only consider the case of one resident trait (i.e. at each time $t\in(s_{k-1},s_k)$ there is exactly one trait $x\in\Xcal$ with $\beta_x(t)=1$) as was already the case in Example \ref{Ex: minimum example}. We introduce the following definitions.
 	\begin{defn}
 		 A trait $x\in\Xcal$ with $f(x+\delta,x)>0$ and $f(x-\delta,x)>0$ is called \emph{local fitness minimum}.
 		 A trait $x\in\Xcal$ with $f(x+\delta,x)<0$ and $f(x-\delta,x)<0$ is called \emph{local fitness maximum}.
 	\end{defn}
 	Now, suppose that the starting trait $x_0\in I$ (and without loss of generality (w.l.o.g.) $x_0\in\Xcal_\delta$ for all $\delta>0$, else consider a shifted grid) is in the interior of $I$ and neither at a local minimum nor maximum. Then there will be one distinguished invading trait (w.l.o.g.~$x_0+\delta$, otherwise it would be $x_0-\delta$) and it will take time $1/f(x_0+\delta,x_0)$ until successful invasion by definition of the times $s_k$. This will trigger a change in the resident trait or it will lead to coexistence. We will assume for now that the population remains monomorphic, so there is always a change in the resident trait, as is usually assumed in order to obtain a trait substitution sequence. However, due to the nature of our model -- more precisely the faster mutation rates -- it may be that the trait $x$ goes to extinction very slowly and becomes fit again after a succession of resident changes. We will exclude this case by imposing the following assumption.
 	
 	\begin{assumption}\label{asmpt: well defined}
 		For any trait $x\in\Xcal_\delta$ with times $t_1<t_2$ such that $\beta_x^\delta(t_1)=1$ and $\beta_x^\delta(t_2)<1$, there exists a time $t_3>t_2$ such that $\beta_x^\delta(t)<1$ for all $t\in[t_2,t_3]$ and $\beta_x^\delta(t_3)=0$. 
 	\end{assumption}
 
 	While this assumption might seem fairly restrictive, it is easily verified that for constant competition $c(x,y)\equiv c$ any rate functions $b$ and $d$ lead to functions $\beta_x^\delta$ which satisfy this criterion.
 	
 	Since we have a sequence of resident traits $x$, we can define a function $g_\delta:[0,\infty)\to\Xcal_\delta$ denoting the unique resident trait (i.e. the trait with exponent $1$). At the times $s_k$ where the resident trait changes and there is no unique resident trait, we use the càdlàg version of $g_\delta$. Without loss of generality, we will assume the function $g_\delta$ to be monotonically increasing remaining consistent with the assumption of the invading trait following $x_0$ to be $x_0+\delta$. Then we can write \[
 	g_\delta(t)=x_0+\max_{t\geq\sum_{k=0}^{i-1}\frac{1}{f(x_0+(k+1)\delta,x_0+k\delta)}} i\delta.
 	\]
 	Note that for smooth fitness functions $f$ the jump times of $g_\delta$ indicating the times until a new mutant indades the population successfully will get longer as we approach a local fitness maximum since the invasion fitness will tend towards $0$. With this definition of the trait substitution sequence $g_\delta$, we can formulate our theorem on the canonical equation of adaptive dynamics with power-law mutation rate in one dimension.
 	
 	\begin{theorem}\label{Thm: CEAD}
 		Assume the set-up introduced above and let the functions $b, d:I\to\R$ and $c:I^2\to\R$ be twice continuously differentiable. Further, let the single starting trait $x_0\in I$ not be a local fitness maximum or minimum. Then, for all $T>0$, the sequence of functions $g_\delta(\cdot/\delta^2)$ converges uniformly as $\delta\to 0$ in the space of càdlàg paths $\mathbb{D}([0,T],\R)$ to the unique function $x\colon [0,T]\to I$ which solves \begin{align}
 		\frac{\mathrm{d}x}{\mathrm{d}t}=\partial_1f(x(t),x(t))\label{eq: CEAD}
 		\end{align}
 		with initial condition $x(0)=x_0$.
 	\end{theorem}
 \begin{proof}
 	Suppose that $y\in I$ is such that at time $s_k$ we have $g_\delta(s_k)=y$ and at the next change of resident trait we have $g_{\delta}(s_{k+1})=y+\delta$. Then we interpolate by setting \[
 	\tilde{g}_\delta(t)=y+\delta f(y+\delta,y)\cdot (t-s_k)\quad\text{for }t\in[s_k,s_{k+1}].
 	\]
 	In particular, this function coincides with $g_\delta$ at the times $s_k$ when the resident trait changes and $g_\delta$ jumps. Since the $L^\infty$ distance of $g_{\delta}(\cdot/\delta^2)$ and $\tilde{g}_\delta(\cdot/\delta^2)$ is bounded by $\delta$, the uniform convergence of the interpolated function and the original function are equivalent. Hence, we will now only consider the function $\tilde{g}_\delta$. Note that $\tilde{g}_\delta$ is a strictly increasing continuous function until the time where $s_{k+1}=\infty$ for some $k\in\N$ and therefore absolutely continuous. We now set $\hat{g}_\delta(t) = \tilde g_\delta(t/\delta^2)$ and compute the right derivative of $\hat{g}_\delta$ at time $\hat{g}_\delta^{-1}(y)$ for $y\in\Xcal^\delta$
 	to be \[
 	\hat g_\delta^{'+}(\hat g_\delta^{-1}(y))=\lim\limits_{h\downarrow 0}\frac{	\hat g_\delta(\hat g_\delta^{-1}(y)+h)-\hat g_\delta(\hat g_\delta^{-1}(y))}{h} = \frac{y+\frac{f(y+\delta,y)}{\delta}\cdot h-y}{h}=\frac{f(y+\delta,y)}{\delta}
 	\]
 	and similarly the left derivative is \[
 	\hat g_\delta^{'-}(\hat g_\delta^{-1}(y))=\frac{f(y,y-\delta)}{\delta}.
 	\]
 	Note that as $\delta\to 0$ these expressions agree and yield $\partial_1f(y,y)=-\partial_2f(y,y)$. For $y\notin\Xcal^\delta$ we find that the function is differentiable with \[
 	\hat g_\delta'(\hat g_\delta^{-1}(y)) = \frac{f(\lfloor y\rfloor_\delta+\delta,\lfloor y\rfloor_\delta)}{\delta}
 	\]
 	where $\lfloor y \rfloor_\delta$ is the largest element of $\delta\Z$ less than or equal to $y$. Since $f$ is continuously differentiable in both variables on a compact set, for any sufficiently small $\varepsilon>0$ we find a $\delta_0$ such that for all $\delta<\delta_0$ we have \[
 	 \left| \frac{f(\lfloor y\rfloor_\delta+\delta,\lfloor y\rfloor_\delta)}{\delta} - \frac{f(y+\delta,y)}{\delta} \right|<\varepsilon
 	\] for all $y\in\Xcal^\delta$. In particular, the limits as $\delta\to 0$ agree up to an error of $\varepsilon$ in a neighbourhood of $y$. Since for any $y\in I$ we can find a sequence $\tilde{\delta}_n\to 0$ such that $y\in\tilde{\delta}_n\Z$, the limits are identical.
 	Hence, taking the limit $\delta\to 0$ of the derivatives, we observe \[
 	\lim\limits_{\delta\to 0}\hat g_\delta'(\hat g_\delta^{-1}(y)) = \partial_1f(y,y) = x'(x^{-1}(y))
 	\]
 	which coincides with the derivative of the solution of the ODE in the statement of the theorem. Hence, the limiting function $g_0(t)$ of $\tilde{g}_\delta(t/\delta^2)$ (and hence of $g_\delta(t/\delta^2)$) solves the ODE \eqref{eq: CEAD}. Existence and uniqueness for a solution of the ODE \eqref{eq: CEAD} follow from $\partial_1f(y,y)$ being Lipschitz on $I$ with Lipschitz constant $L\geq 0$. By uniqueness of the solution of the ODE, we have $g_0=x$. To show uniform convergence, we see that for every $\delta>0$ we have \[
 	\hat{g}_\delta' = \partial_1f(\hat{g}_\delta,\hat{g}_\delta)+r(\hat{g}_\delta,\delta)
 	\]
 	almost everywhere, where $r(y,\delta)$ is a bounded function in $y$ and converges uniformly to $0$ as $\delta\to 0$ by uniform convergence of the differential quotient due to compactness of $I$. Thus, \begin{align*}
 		|\hat{g}_\delta(t)-x(t)|&\leq \int_{0}^{t} |\partial_1f(\hat{g}_\delta(s),\hat{g}_\delta(s))-\partial_1f(x(s),x(s))|+|r(\hat{g}_\delta(s),\delta)|\ \mathrm{d}s\\
 		&\leq \int_{0}^t L|\hat{g}_\delta(s)-x(s)| \ \mathrm{d}s + M(\delta)t,
 	\end{align*}
 where $M$ is a function such that $|r(y,\delta)|\leq M(\delta)$ and $M(\delta)\to 0$ as $\delta\to 0$.
    Uniform convergence follows using Gronwall's inequality.
 \end{proof}
 
 \begin{remark}
 	We have assumed throughout that the function $g_\delta$ and hence $\tilde{g}_\delta$ are monotonically increasing. Our result also applies for monotonically decreasing functions by considering the invading traits to be $x-\delta$ instead of $x+\delta$. 
 \end{remark}
\begin{remark}\label{Rem: comparison}
	In \cite{CM11}, the canonical equation reads \[
	\frac{dx}{dt}=\int_{\R^\ell}h[h\cdot\nabla_1g(x,x)]_+m(x,h)dh.
	\]
	Here, $g$ is a function related to the fitness $f$ and $m$ is the mutation kernel. In one dimension ($\ell=1$), this would mean in our setting that $m(x,h)=\delta_1(h)$ is the point measure on $1$ and $g$ is replaced by $f$. Then the equation reads \[
	\frac{dx}{dt} = \nabla_1f(x,x),
	\]
	which is the result that we have obtained. In contrast to this result, we require slightly more regularity on the functions which determine the fitness function $f$, because we want the derivative of $f$ with respect to the first component to be Lipschitz.
\end{remark}

\subsection{Evolutionary branching in one dimension}
Now that we have shown the result on the canonical equation, we can also consider evolutionary singularities, i.e. stationary points of the ODE \eqref{eq: CEAD}. Such points are also called evolutionary singular strategies in the literature. Here, these are the points $x\in I$ satisfying $\partial_1 f(x,x)=0$. We will show that the same criterion for evolutionary branching in the sense of Champagnat and Méléard \cite{CM11} will apply. We recall the definition and necessary assumptions here. \begin{defn}
	Let $x^*\in\Xcal$ be an evolutionary singularity, i.e. an equilibrium of equation \eqref{eq: CEAD}. We call \emph{$\eta$-branching} the event that satisfies \begin{itemize}
		\item There is a time $t_1$ such that there is a unique resident trait $x_0$ at time $t_1$ contained in the interval $[x^*-\eta,x^*+\eta]$.
		\item There is a time $t_2$ such that there is coexistence between exactly two traits distant more than $\eta/2$ at time $t_2$.
		\item In the time interval $[t_1,t_2]$, there are always at most two (coexisting) resident traits with increasing distance over time between the traits.
	\end{itemize}
\end{defn}

Since we are now dealing with coexistence of traits, we need to \emph{assume} that the corresponding Lotka-Volterra system has a unique stable equilibrium.
 \begin{assumption}\label{asmpt: stability}
	Suppose that $\mathbf{v}\subseteq\Xcal$ is a set of $k$ coexisting traits. Then we assume that for any mutant trait $y\notin\mathbf{v}$ with $f(y,\mathbf{v})>0$ the solution of the corresponding mutation free Lotka-Volterra system for the set of traits $\mathbf{v}\cup\{y\}$ converges to a unique equilibrium $n^*$ for any starting value in a sufficiently small neighbourhood $\Ucal\subseteq\R_{\geq 0}^{k+1}$ of $(\bar{n}(\mathbf{v}),0)$. Further assume that for all traits $x_j\in\mathbf{v}$ whose equilibrium coordinate in $n^*$ is $0$, we have $f(x_j,\mathbf{v}^*)<0$, where $\mathbf{v}^*$ is the collection of traits whose coordinate in $n^*$ is strictly positive.
\end{assumption}

With the notion of evolutionary branching in the sense of $\eta$-branching and assumption on the stability of equilibrium points, we can state our result which coincides with the one in \cite{CM11}.

\begin{prop}
	Assume the set-up from Theorem \ref{Thm: CEAD} and Assumption \ref{asmpt: stability}. In addition, assume the functions $b$ and $d$ to be three times and $c$ to be four times continuously differentiable. Further, let the functions $\beta_x^\delta$ from \eqref{eq: PES} start with the unique resident trait $x_0\in\Xcal$ and assume that the sequence of resident traits converges towards an evolutionary singularity $x^*\in\Xcal$ in the interior of $\Xcal$. Lastly assume that this singularity satisfies \begin{align*}
		\partial_{22}f(x^*,x^*)>\partial_{11}f(x^*,x^*)\quad\text{and}\quad \partial_{22}f(x^*,x^*)+\partial_{11}f(x^*,x^*)\neq 0.
	\end{align*}
	Then, for all sufficiently small $\eta>0$ there exists $\delta_0>0$ such that for all $\delta\leq\delta_0$ \begin{itemize}
		\item if $\partial_{11}f(x^*,x^*)>0$, there is  $\eta$-branching at $x^*$.
		\item if $\partial_{11}f(x^*,x^*)<0$, there is no $\eta$-branching at $x^*$.
	\end{itemize}
\end{prop}
\begin{proof}
	The result on evolutionary branching is an entirely analytic result using properties of the invasion fitness and doe not depend on a particular choice of a mutation kernel or mutation rate besides the mutation kernel having mass on the positive and negative real numbers. Therefore, we can apply the same proof as in \cite[Theorem 4.10.]{CM11}.
\end{proof}

\subsection{Examples}
	We want to give some examples for convergence to the CEAD in this setting as well as for evolutionary branching. To this end, we consider the example discussed in \cite{CM11} and proposed in \cite{DD99}. Here, the interval is $I=[-2,2]$, the birth rate is $b(x)=\exp(-x^2/(2\sigma_b^2))$, death rate $d(x)=0$ and competition kernel $c(x,y)=\exp(-(x-y)^2/(2\sigma_c^2))$. We expect the population to evolve into the reproduction optimum at $0$ for all choices of $\sigma_b^2,\sigma_c^2>0$. If the concentration of the competition kernel is higher than that of the birth rate (i.e. $\sigma_b^2>\sigma_c^2$), we expect to see evolutionary branching since the loss in reproduction away from $0$ is compensated by a reduction in death from competition. This is also confirmed by computing the derivatives from the branching criterion.\\
	
	Firstly, we want to consider the approximation of the function $g_\delta$ by the canonical equation as predicted by Theorem \ref{Thm: CEAD}. In  this setting, let $\delta = 0.01$ and $T=10$. Further, let $\sigma_b^2=1.8$ and $\sigma_c^2=1.5$. Then we obtain the graphs for $x(t)$ and $g_\delta(t/\delta^2)$ as shown in Figure \ref{Fig: Comparison 1dim CEAD}. \begin{figure}[htbp]
		\begin{tikzpicture}
			\node (img1)  {\includegraphics[width=0.5\linewidth]{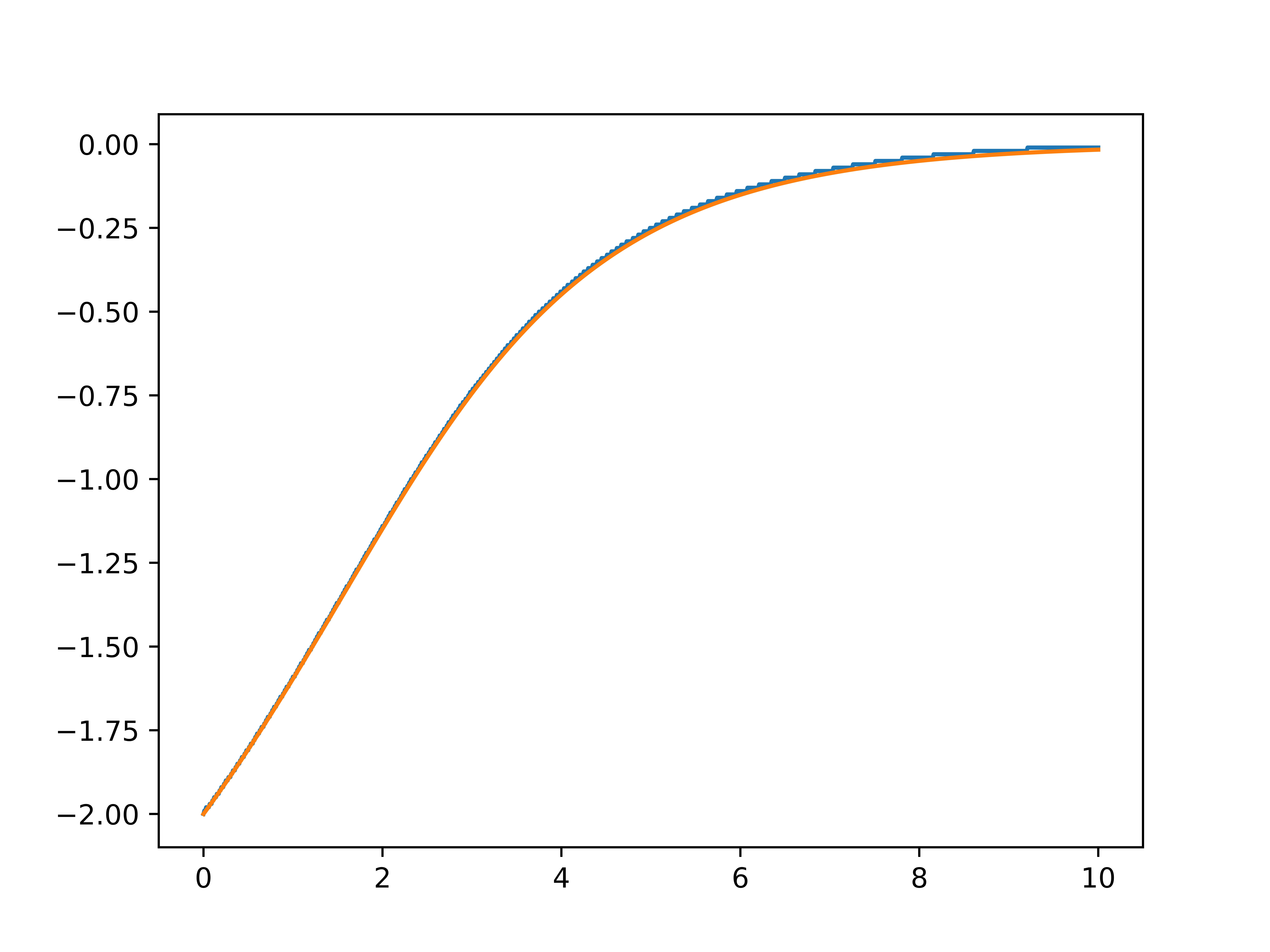}};
			\node[below=of img1, node distance=0cm, yshift=1.5cm, xshift=0.0cm] {\footnotesize Time $t$};
			\node[left=of img1, node distance=0cm, rotate=90, anchor=center,yshift=-1.3cm] {\footnotesize Trait};
		\end{tikzpicture}
	\caption{The solution $x(t)$ (orange) of the canonical equation and the function $g_\delta(t/\delta^2)$ (blue) for $\delta=0.01$ both started at $x_0=-2$.}
	\label{Fig: Comparison 1dim CEAD}
	\end{figure}
	As we would expect, the difference between $x(t)$ and $g_\delta(t/\delta^2)$ is bounded by $\delta$ since this is the furthest a point can be from a (directed) $\delta$-grid in one dimension. \\
	
	Secondly, we want to consider an illustration of evolutionary branching. Our choice of parameters satisfies the branching criterion and as shown in our consideration for the canonical equation, we reach a neighbourhood of the singularity in (real time) $10^5$ time steps for $\delta = 0.01$. Therefore, we expect evolutionary branching to occur around this time and we expect the branching to be symmetric around $0$ by the deterministic nature of our functions. Indeed, we find the behaviour shown in the left image of Figure \ref{Fig: 1dim branching}.
	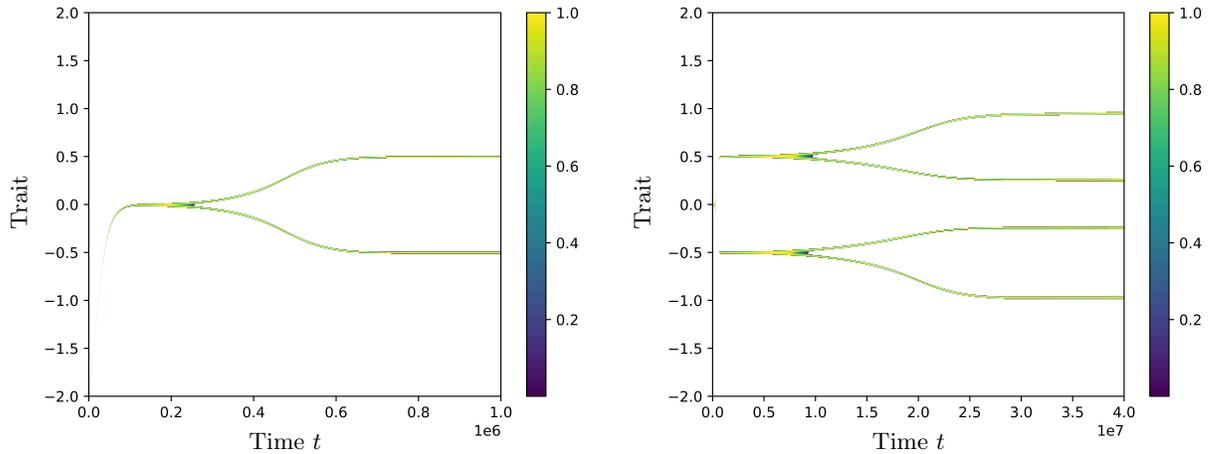
\begin{figure}[htbp]
		\begin{minipage}{0.48\linewidth}
			\hspace{0.8cm}\begin{tikzpicture}
				\node (img1)  {\includegraphics[width=1.08\linewidth]{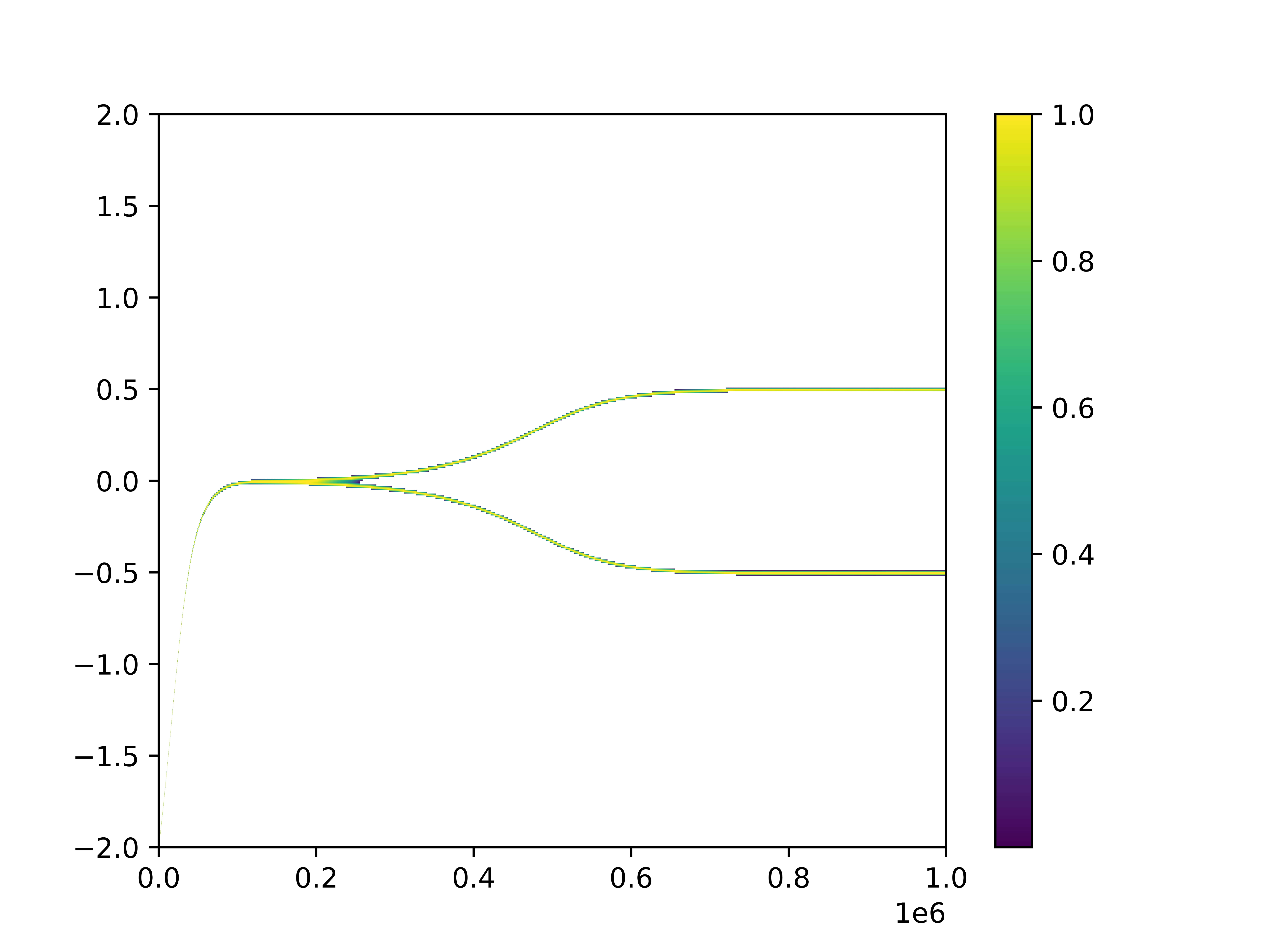}};
				\node[below=of img1, node distance=0cm, yshift=1.5cm, xshift=-0.7cm] {\footnotesize Time $t$ };
				\node[left=of img1, node distance=0cm, rotate=90, anchor=center,yshift=-1.3cm] {\footnotesize Trait};
			\end{tikzpicture}
		\end{minipage}
	\begin{minipage}{0.48\linewidth}
		\hspace{0.8cm}\begin{tikzpicture}
			\node (img1)  {\includegraphics[width=1.08\linewidth]{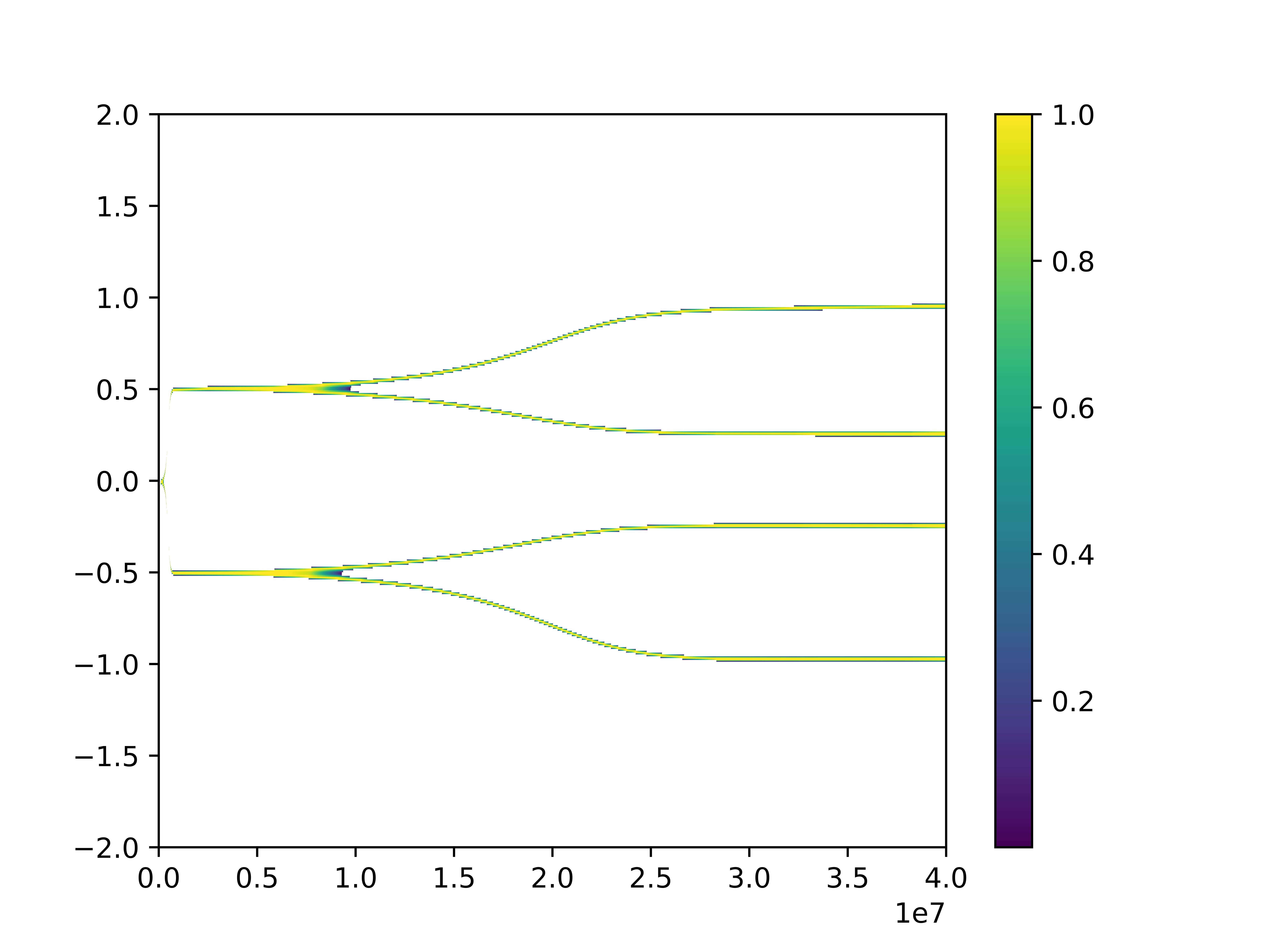}};
			\node[below=of img1, node distance=0cm, yshift=1.5cm, xshift=-0.7cm] {\footnotesize Time $t$ };
			\node[left=of img1, node distance=0cm, rotate=90, anchor=center,yshift=-1.3cm] {\footnotesize Trait};
		\end{tikzpicture}
	\end{minipage}
		\caption{The evolution of $\beta_x^\delta$ over time. Colour indicates the value of $\beta_x^\delta$ at a given time $t$ for a trait $x$. For better legibility of the figure, we assign white to represent traits with $\beta_x^\delta=0$. The process is started with $\beta_{-2}^\delta(0) = 1$ and $\beta_x^\delta(0) = 0$ for $x\neq -2$. The time horizon is $T=10^6$ (left) and $T=4\times 10^7$ (right).}
		\label{Fig: 1dim branching}
	\end{figure}
	The resident trait reaches a neighbourhood of the evolutionary singularity $0$ in a short period relative to the time it takes for the branching to reach its end. This is due to the fact that at first the invading traits only need to compete against a single resident trait and have both a higher rate of reproduction and a lower exposure to competition compared to individuals of the resident trait. During the branching phase however the mutants advantage is reduced to the difference between loss in reproduction and reduced competition with the resident traits. Hence it takes longer for mutants to be successful during this phase. On an even longer time scale, we see two subsequent branching events (see right image of Figure \ref{Fig: 1dim branching}). These are not symmetric anymore since the branches in the middle experience much higher competition than those on the outsides.

\section{The canonical equation  in two dimensions}\label{Sec: 2dim}
	We will now concern ourselves with the corresponding result of Theorem \ref{Thm: CEAD} for a two dimensional trait space. Higher dimensional trait spaces are necessary to investigate the effects of different components of a trait such as reproduction rate, tolerance of competitive pressure, dormancy or horizontal transfer as was done for example in \cite{BPT23}. While trade-offs between these features can also be understood in a one-dimensional trait space as in \cite{CMT21} or \cite{DD99}, this requires a simplification of the interactions. For our purpose, let now $\Xcal = (I\times J)\cap\delta\Z^2$ be the grid of size $\delta$ on a rectangle formed by the intervals $I,J\subseteq\R$. Concerning mutations, we will only allow mutations to the direct neighbours i.e.~we see mutations from trait $(x_1,x_2)$ to $(x_1\pm\delta,x_2)$ or $(x_1,x_2\pm\delta)$ with equal probability $K^{-\alpha}/4$ at birth. The birth function $b(x)$, death rate $d(x)$ and competition kernel $c(x,y)$ are again as before but now depending on the two dimensional traits $x,y\in\Xcal$.\\
	
	Since the result \cite[Theorem 2.2]{CKS21} (formulated in this paper in Theorem \ref{Thm: CKS21}) holds for a general trait space, we can apply their result to this new situation as well and obtain a family of piecewise affine functions $\beta$. Again, we will let $\alpha\to 1$ to obtain the functions $\beta_x^\delta$ as detailed in equation \eqref{eq: PES}. A key component to the derivation of the canonical equation in one dimension was to assume that we are neither at a local fitness maximum nor at a local fitness minimum so that we know in which direction the resident trait will evolve in trait space. However, in a two dimensional trait space, there may still be more than one invading trait even outside of a fitness minimum. A generic example is that mutations both along the first component and along the second component have an evolutionary advantage over their parental trait. This could lead to situations in which again one of the invading traits becomes resident and hence affects the fitness of the secondary invading trait. We want to avoid such situations in order to have control over the sequence of resident traits by imposing another assumption.
	
	\begin{assumption}\label{asmpt: simplification}
		We assume that for any of the four possible mutant traits $y_1,\ldots, y_4$ of the parental trait $x$, the trait with the highest fitness will become resident and all the remaining mutant traits have a negative fitness against this new resident and will go to extinction, even when there are further changes in the resident trait prior to extinction. If we assume as an example $f(y_i,x)>f(y_{i+1},x)$ for $i=1,2,3$, then we assume $y_1$ to become the resident trait following $x$ and $f(y_{i},y_1)<0$ for $i=2,3,4$. However, since further changes in the resident trait could lead to a positive fitness of the mutant traits $y_2,y_3,y_4$, we assume that these traits go to extinction instead.
	\end{assumption}
	Assumption \ref{asmpt: simplification} is the analogue of Assumption \ref{asmpt: well defined}, but this time we do not only need to control traits which are resident at one point and then have decreasing fitness, but we also need to control all the mutant traits which reach strictly positive exponents. While this is quite a strong restriction, it prohibits ``jumps'' in trait space, that is, we avoid the situation in which the resident trait jumps by more than $\delta$ in either of the coordinates. However, it also appears natural to assume that usually the trait with the highest fitness out of a selection of possible mutants will also be fit against the competing traits.\\
	
	In the following, we again want to concern ourselves with the sequence of resident traits over time. For this, observe that the fitness function $f$ dictates the next invading trait. In particular, the gradient with respect to the first component tells us, in which direction the new invading trait lies. If the first coordinate of the gradient $\nabla_1f(x,x)$ at a point $x\in\Xcal$ is larger than the second coordinate, we will see the next resident trait in the first coordinate. As long as the coordinates of $\nabla_1 f$ are unequal, we know in which direction the next resident trait lies. However, there may be a set where the coordinates of $\nabla_1 f$ are equal. To deal with this case, we introduce the notion of an \emph{attractive curve} given the fitness function $f$. \begin{defn}\label{Defn: attractive}
		For a vector $x\in\R^\ell$ we write $[x]_i$ for the $i$-th component of $x$.
		Consider the set $\Mcal:=\{(x,g(x))\mid x\in D\}$ for some domain $D\subseteq\R$ and a differentiable and monotone function $g:D\to\R$, i.e.~$\Mcal$ is the graph of $g$ on the domain $D$. Then we call the set $\Mcal$ \emph{attractive for $f$} if one of the following holds: \begin{itemize}
			\item The function $g$ is monotonically increasing and for any sufficiently small $\varepsilon>0$ we have \[
			\pm[\nabla_1 f((x\mp\varepsilon,g(x)),x\mp\varepsilon,g(x))]_1>|[\nabla_1 f((x\mp\varepsilon,g(x)),x\mp\varepsilon,g(x))]_2|\]
			 as well as \[\pm[\nabla_1 f((x\pm\varepsilon,g(x)),x\pm\varepsilon,g(x))]_2>|[\nabla_1 f((x\pm\varepsilon,g(x)),x\pm\varepsilon,g(x))]_1|.\] 
			\item The function $g$ is monotonically decreasing and for any sufficiently small $\varepsilon>0$ we have \[
			\pm[\nabla_1 f((x\mp\varepsilon,g(x)),x\mp\varepsilon,g(x))]_1>|[\nabla_1 f((x\mp\varepsilon,g(x)),x\mp\varepsilon,g(x))]_2|\]
			as well as \[\mp[\nabla_1 f((x\pm\varepsilon,g(x)),x\pm\varepsilon,g(x))]_2>|[\nabla_1 f((x\pm\varepsilon,g(x)),x\pm\varepsilon,g(x))]_1|.\] 
	\end{itemize}
	The signs in the individual displays are fixed by the leading sign on the left hand side of the inequalities, i.e.~in each display either all signs are from the top row or from the bottom row. Also, for the two displays in each case, the choice of the row of signs must be the same. As an example, if we choose ``+'' as the leading sign on the left hand side of the first display of the second case, then we need to choose ``-'' as leading sign in the second display.
	\end{defn}
	From this definition we see that we would need to distinguish many cases to formulate our results in all generality. For simplicity, we will assume the first case of the definition to be satisfied with positive leading signs on the left hand side of the displays. The definition is best understood with a simple image. We illustrate attractiveness for the line determined by $\Mcal = \{(x,x)\mid x\in[0,1]\}$ in Figure \ref{Fig: example stable curve} and a generic fitness function $f$. Intuitively, the function showing the current resident trait will see a movement through trait space in the direction of the component of $\nabla_1f$ which has the largest absolute value. A curve then is attractive if we have a drift towards the curve given by $\Mcal$. In the example, the line $\Mcal$ is attractive if for all $x$ in the yellow area, $[\nabla_1f(x,x)]_1>[\nabla_1f(x,x)]_2\geq0$ and the same holds true for the second coordinate for all $x$ in the black area. This would see a movement from the bottom left to the top right along the curve because the yellow area indicates the resident trait to evolve towards the right while the black area promotes new resident traits which have a larger second component. Of course, we could also go in the reverse direction, which is the first case of the definition with leading negative signs. If these conditions do not hold, we will see the resident trait move away from the curve in a straight line.
	
	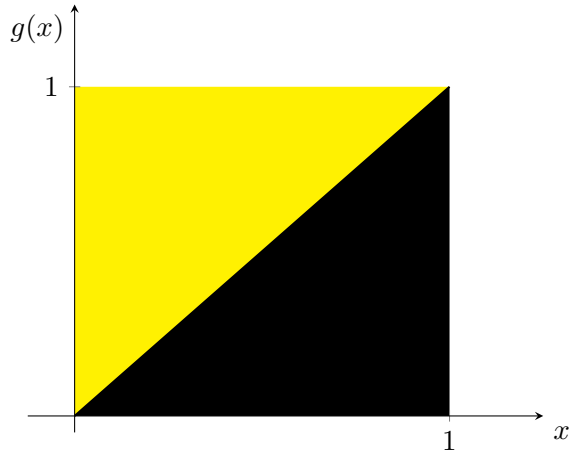
\begin{figure}[htbp]
		\begin{tikzpicture}[declare function={a=0.5;b=3;f(\x)=\x;}]
			\begin{axis}[axis lines=middle,axis on top,xlabel=$x$,ylabel=$g(x)$,
				xmin=-0.5,xmax=5,ymin=-0.2,ymax=5,xtick={4},xticklabels={$1$},
				ytick={4},yticklabels={$1$},
				every axis x label/.style={at={(current axis.right of origin)},anchor=north west},
				every axis y label/.style={at={(current axis.above origin)},anchor=north east}]
				\addplot[name path=A,black,thick,domain=0:4,samples=100] {f(x)};
				\path[name path=B] (0,4) -- (4,4);
				\path[name path=C] (\pgfkeysvalueof{/pgfplots/xmin},0) -- (\pgfkeysvalueof{/pgfplots/xmax},0);
				\addplot [yellow] fill between [
				of=A and B,soft clip={domain=0:4},
				];
				\addplot [black] fill between [
				of=A and C,soft clip={domain=0:4},
				];
				;
			\end{axis}
		\end{tikzpicture}
		\caption{An example of a set $\Mcal$ dividing the square $[0,4]^2$ into two parts. The yellow area represents a drift to the right, the black area represents an upward drift.}
		\label{Fig: example stable curve}
	\end{figure}
	We now have all the necessary prerequisites to formulate our theorem for the canonical equation in two dimensions. In order to keep the theorem legible, we will only formulate it for the case in which both coordinates of $\nabla_1f$ are non-negative. For other cases, suitable absolute values and leading negative signs need to be introduced where appropriate but they will yield the same result.  \begin{theorem}\label{Thm: 2dim CEAD}
		Denote the function tracing the unique resident trait at time $t$ by $h_\delta(t)$, where at times of a change in the resident trait, we choose $h_\delta$ to take the value of the new resident trait so that $h_\delta$ is càdlàg.
		
		 Suppose that Assumption \ref{asmpt: simplification} is satisfied and assume the functions $b,d\colon I\times J\to\R$ and $c\colon (I\times J)^2\to\R$ to be twice continuously differentiable. Define the fitness function $f:(I\times J)^2\to\R$ as in \eqref{eq: fitness}. Further, let $x_0\in I\times J$ be such that $\nabla_1f(x_0,x_0)$ is not a multiple of $(1,1)$. Lastly, assume that the gradient $\nabla_1f(x,x)$ is a coordinate-wise non-negative vector for all $x\in I\times J$ and the set $\Mcal = \{(x\in I\times J\mid [\nabla_1f(x,x)]_1 = [\nabla_1f(x,x)]_2\}$ can be written as a finite union of graphs of monotone and continuously differentiable functions $g_i:\R\to\R$.
		 
		  Then, for all $T>0$, the sequence of functions $h_\delta(\cdot/\delta^2)$ converges uniformly as $\delta\to 0$ in the space of càdlàg functions $\D([0,T],\R^2)$ to the unique function $x \colon [0,T]\to\R^2$ with $x(0)=x_0$ and which is characterized piecewise in the following way: \begin{enumerate}
			\item If $[\nabla_1f(x_0,x_0)]_1>[\nabla_1f(x_0,x_0)]_2>0$, then there exists a time $s_1\in (0,\infty]$ such that $[\nabla_1f(x(s_1),x(s_1))]_1=[\nabla_1f(x(s_1),x(s_1))]_2$ and for all $t<s_1$ we have \[
			\dot{x}(t)=\begin{pmatrix}
				[\nabla_1f(x(t),x(t))]_1\\ 0
			\end{pmatrix}
			\]
			and if $[\nabla_1f(x_0,x_0)]_2>[\nabla_1f(x_0,x_0)]_1>0$, then we have \[
			\dot{x}(t)=\begin{pmatrix}
				0\\ [\nabla_1f(x(t),x(t))]_2
			\end{pmatrix}.
			\]
			\item After time $s_1$, \begin{itemize}
				\item if the set $\Mcal$ is attractive for $f$, the function values of $x$ will trace the curve given by the points in the set $\Mcal$, i.e. $x(t)\in\Mcal$ for $t\geq s_1$ sufficiently small. The speed at which the function $x$ moves along this curve is given by \[
				\frac{	[\nabla_1f(x(t),x(t))]_1}{\sin(\angle x(t))+\cos(\angle x(t))} = \frac{	[\nabla_1f(x(t),x(t))]_2}{\sin(\angle x(t))+\cos(\angle x(t))},
				\] where $\angle x$  denotes the angle of the tangent line of the curve $\Mcal$ at $x\in\Mcal$ to the positive real axis modulo $\pi/2$. Hence, if $T_M(x)$ denotes the normalized tangent vector with appropriate orientation of $M$ at $x$, then we have \begin{align}
				\dot{x}(t)=\frac{	[\nabla_1f(x(t),x(t))]_1}{\sin(\angle x(t))+\cos(\angle x(t))}\cdot T_M(x(t)). \label{eq: 2dimCEAD}
				\end{align} This occurs until a time $s_2\in (s_1,\infty]$, where the tangent line of the set $\Mcal$ is a horizontal or vertical line in trait space or the edge of the trait space is reached.
				\item if the set $\Mcal$ is not attractive for $f$, then the function $x$ will continue in a straight line in the direction of the initially non-dominating coordinate at the appropriate speed given in (1).
			\end{itemize} 
			\item If the tangent line at time $s_2$ is horizontal, then the function $x$ will continue horizontally as described in (1). If the tangent is vertical, then we will see the analogue continuation of $x$ in the direction of the second component. If we have reached the boundary of the trait space, then we will see $x$ continue also in only one dimension along the coordinate which is not on the boundary.
			\item The phases (2) and (3) alternate until eventually $\nabla_1f(x(t),x(t))=0$.
		\end{enumerate}
	\end{theorem}
\begin{remark}
	Theorem \ref{Thm: 2dim CEAD} should also be true without Assumption \ref{asmpt: simplification}. Intuitively, due to smoothness of $f$ in both arguments, we will not see ``large'' jumps in trait space for $\delta$ sufficiently small. However, it is difficult to quantify this and hence we prove the theorem with this additional assumption.\\

\end{remark}
\begin{proof}[Proof of Theorem \ref{Thm: 2dim CEAD}]
	We give a proof of the piecewise convergence by considering the different phases. In the beginning, due to the assumption on $x_0$, the gradient $\nabla_1f(x_0,x_0)$ has unequal coordinates. W.l.o.g.~we assume the first coordinate to have larger absolute value. Then, by Assumption \ref{asmpt: simplification} we will see one dimensional evolution in trait space and hence can apply Theorem \ref{Thm: CEAD}. By the intermediate value theorem, there may exist a distinguished point $x_0+(a,0)$ for some $a>0$ at which the coordinates of $\nabla_1f$ become equal. This point is reached by the function $x$ at a time $s_1>0$. Beyond this point, we will start to see successful invasions from mutations in the direction of the second coordinate and hence we cannot reduce our arguments to the one dimensional case.\\
	
	We find that the curve described by the set $\Mcal$ divides the trait space into distinct regions where on either side we have $[\nabla_1f(y,y)]_1>[\nabla_1f(y,y)]_2$ or vice versa. Assuming without loss of generality the former case to apply, this implies for sufficiently small $\delta>0$ that \[
	 f((y_1+\delta,y_2),(y_1,y_2))>f((y_1,y_2+\delta),(y_1,y_2)).
	\]
	With Assumption \ref{asmpt: simplification} this shows that the resident trait will evolve in the direction of the first component until the line $\Mcal$ is crossed again and  since $\Mcal$ is attractive for $f$, the inequality is reversed. In particular, we see that the resident trait will not leave a $\delta$-neighbourhood of the set $\Mcal$ during this phase. However, if the tangent line of $\Mcal$ becomes a horizontal or vertical line, then this crossing of $\Mcal$ does not occur anymore and we will see evolution in one coordinate again for which the proof of case $(1)$ applies. Also, if $\Mcal$ is not attractive for $f$, we will see only one-dimensional evolution by definition of attractiveness but now in the direction of the second coordinate.\\
	
	To give proof for the speed at which the limiting function moves along the curve given by the set $\Mcal$, we first consider the case of a straight line. Suppose that $\Mcal = \{(x,ax+b)\in\R^2\mid x\geq0 \}$ for some fixed $a\geq 0$ and $b\in\R$. Further, let $\varphi\in[0,\pi/2)$ denote the angle between the line $\Mcal$ and the positive real axis. Since we assume $\Mcal$ to be attractive for $f$, w.l.o.g.~we assume the first case of Definition \ref{Defn: attractive} to hold. Now, consider a point $x_1\in\Mcal$ and suppose that there is a sequence of $\delta$-grids such that we cross $\Mcal$ at $x_1$. W.l.o.g., $h_\delta$ crosses $\Mcal$ at $x_1$ vertically. Let $\delta_n\to 0$ be a corresponding sequence such that $[x_1]_1 \in\delta_n\Z$. To simplify notation, we consider a fixed $\delta>0$ as an element of this sequence. We calculate the time until the next crossing of $\Mcal$ occurs and subsequently we obtain the speed at which we trace $\Mcal$.\\
	
	For this, we again consider the interpolated version of the function $h_\delta$ which we call $\tilde{h}_\delta$. Recall the notation $\lfloor y\rfloor_\delta$ to mean the largest element of $\delta\Z$ which is less than or equal to $y\in\R$. Then there is a point $x_{1,\delta} = ([x_1]_1,\lfloor [x_1]_2 \rfloor_\delta)$ on the grid such that there is a time $t_0$ with $h_\delta(t_0)=x_{1,\delta}$ and a time \[
	t_2 = t_0 + \frac{1}{f(x_{1,\delta}+\delta\mathbf{e}_2, x_{1,\delta})} 
	\]
	with $h_\delta(t_2)=x_{1,\delta}+\delta\mathbf{e}_2$, where $\mathbf{e}_i$ denotes the  $i$-th unit vector. By construction, there is a time $t_1\in[t_0,t_2)$ such that the interpolation of $h_\delta$ on this time interval satisfies $\tilde{h}_\delta(t_1)=x_1$.\\
	
	Now, let $\varepsilon>0$ be such that $x_2 := x_{1,\delta}+\delta\mathbf{e}_2=x_1+\varepsilon\sin(\varphi)\mathbf{e}_2$. Note that $\varepsilon \leq\delta/\sin(\varphi)$ is dependent on $\delta$ and with $\delta\to 0$ we also have $\varepsilon\to0$. By our assumption on attractivity of $\Mcal$ for $f$ we then find that there is a time $t_3>t_2$ and a point $x_3\in\Mcal$ such that \[
	x_3=x_1+\varepsilon\left(\cos(\varphi), \sin(\varphi)\right)
	\]
	and $\tilde{h}(t_3)=x_3$. Note that $\lVert x_3-x_1\rVert_2 = \varepsilon$.\\
	
	We now calculate the time difference $t_3-t_1$ to determine the speed $\lVert x_3-x_1\rVert_2 / (t_3-t_1)$. Using the interpolation function $\tilde{h}$ we find that the time $t_2-t_1$ satisfies \[
	t_2-t_1 = \frac{\varepsilon\sin(\varphi)}{\delta f(x_{1,\delta}+\delta\mathbf{e}_2, x_{1,\delta})}.
	\]
	
	Next, we cover a distance of $\varepsilon\cos(\varphi)$ units horizontally which sees a total of $\lfloor\varepsilon\cos(\varphi)/\delta\rfloor$ steps on the grid in addition to an interpolation step. Hence, this takes a time of \[
	t_3-t_2 = \frac{\varepsilon \cos(\varphi) - \delta\lfloor\varepsilon\cos(\varphi)/\delta\rfloor}{\delta f((\lfloor [x_3]_1\rfloor_\delta+\delta, [x_3]_2),(\lfloor [x_3]_1\rfloor_\delta, [x_3]_2) )} +\sum_{k=1}^{\lfloor\varepsilon\cos(\varphi)/\delta\rfloor}\frac{\delta}{\delta f(x_2+k\delta\mathbf{e}_2, x_2+(k-1)\delta\mathbf{e}_2)}.
	\]
	Accelerating time by $1/\delta^2$ and with a slight abuse of notation, we see that \[
	\frac{t_2-t_1}{\delta^2}=\frac{\varepsilon\sin(\varphi)}{\delta f(x_{1,\delta}+\delta\mathbf{e}_2, x_{1,\delta})} \quad\Longleftrightarrow\quad t_2-t_1 = \frac{\delta\varepsilon\sin(\varphi)}{f(x_{1,\delta}+\delta\mathbf{e}_2, x_{1,\delta})}
	\]
	and similarly for $(t_3-t_2)/\delta^2$ we get \[
		t_3-t_2 = \frac{\delta(\varepsilon \cos(\varphi) - \delta\lfloor\varepsilon\cos(\varphi)/\delta\rfloor)}{f((\lfloor [x_3]_1\rfloor_\delta+\delta, [x_3]_2),(\lfloor [x_3]_1\rfloor_\delta, [x_3]_2) )} +\sum_{k=1}^{\lfloor\varepsilon\cos(\varphi)/\delta\rfloor}\frac{\delta^2}{f(x_2+k\delta\mathbf{e}_2, x_2+(k-1)\delta\mathbf{e}_2)}.
	\]
	Now, we find that the quotients are approximately the reciprocal of the derivative of $f$. Hence, we write \[
	t_2-t_1 = \varepsilon\sin(\varphi)\left(\frac{1}{[\nabla_1f(x_1,x_1)]_2} + r_1(\delta)\right)
	\]
	for some error function $r_1(\delta)$ which tends to $0$ with $\delta\to0$. Using the continuous differentiability of $f$, we can perform a similar substitution for the time step \[
	t_3-t_2 = (\varepsilon \cos(\varphi) - \delta\lfloor\varepsilon\cos(\varphi)/\delta\rfloor)\cdot\left(\frac{1}{[\nabla_1f(x_3,x_3)]_1}+r_3(\delta)\right) +\sum_{k=1}^{\lfloor\varepsilon\cos(\varphi)/\delta\rfloor}\left(\frac{\delta}{[\nabla_1f(x_2,x_2)]_1} +\delta r_{2,k}(\delta)\right).
	\]
	Since the sum now only has constant terms (except for the error terms $r_{2,k}$), we replace it by setting  \[
	t_3-t_2 = (\varepsilon \cos(\varphi) - \delta\lfloor\varepsilon\cos(\varphi)/\delta\rfloor)\cdot\left(\frac{1}{[\nabla_1f(x_3,x_3)]_1}+r_3(\delta)\right) +\delta\lfloor\varepsilon\cos(\varphi)/\delta\rfloor\left(\frac{1}{[\nabla_1f(x_2,x_2)]_1}+ r_{2}(\delta)\right)
	\]
	for an appropriate error function $r_2$. This is possible since the functions $r_{2,k}$ are uniformly bounded. By continuity of $\nabla_1f$ we again find that -- with the introduction of another error term -- we can replace the points $x_3$ and $x_2$ with $x_1$ to obtain \begin{align*}
	t_3-t_2 = \frac{\varepsilon \cos(\varphi) }{[\nabla_1f(x_1,x_1)]_1} &+ \delta\lfloor\varepsilon\cos(\varphi)/\delta\rfloor r_2(\delta)+ (\varepsilon \cos(\varphi) - \delta\lfloor\varepsilon\cos(\varphi)/\delta\rfloor) r_3(\delta)\\& + \varepsilon \cos(\varphi) r_4(\delta).
	\end{align*}
	Hence, the total time taken to go from $x_1$ to $x_3$ can be written as \[
	t_3-t_1 = \frac{\varepsilon\sin(\varphi)}{[\nabla_1f(x_1,x_1)]_2} + \frac{\varepsilon \cos(\varphi) }{[\nabla_1f(x_1,x_1)]_1}+  \varepsilon\sin(\varphi) r_1(\delta)+\varepsilon\cos(\varphi)r_5(\delta)
	\]
	for some error function $r_5(\delta)$.
	By definition of the set $\Mcal$, at $x_1\in\Mcal$ we have $[\nabla_1f(x_1,x_1)]_2=[\nabla_1f(x_1,x_1)]_1$. Therefore, the speed at which $h_\delta$ travels from $x_1$ to $x_3$ is now determined by \begin{equation}\begin{aligned}
		\frac{\lVert x_3-x_1\rVert_2}{t_3-t_1}&=\frac{\varepsilon}{\frac{\varepsilon\sin(\varphi)}{[\nabla_1f(x_1,x_1)]_1} + \frac{\varepsilon \cos(\varphi) }{[\nabla_1f(x_1,x_1)]_1}+  \varepsilon\sin(\varphi) r_1(\delta)+\varepsilon\cos(\varphi)r_5(\delta)}\\
		&=\frac{[\nabla_1f(x_1,x_1)]_1}{\sin(\varphi) + \cos(\varphi)+  [\nabla_1f(x_1,x_1)]_1\sin(\varphi) r_1(\delta)+[\nabla_1f(x_1,x_1)]_1\cos(\varphi)r_5(\delta)}\\
		&\xrightarrow{\delta\to 0}\frac{[\nabla_1f(x_1,x_1)]_1}{\sin(\varphi) + \cos(\varphi)}.
	\end{aligned}\label{eq: limit}
	\end{equation}
	In conclusion, the infinitesimal speed at which the limiting function moves along the straight line $\Mcal$ is given at the point $x\in\Mcal$ by \[
	\frac{[\nabla_1f(x,x)]_1}{\sin(\varphi) + \cos(\varphi)}.
	\]
	Now, for general graphs $\Mcal$, observe that performing the same calculation via the approximation of the graph with the tangent line incurs an error which is of the order $\varepsilon^2$ due to our assumption on the differentiability of the corresponding functions $g_i$. As we let $\delta\to 0$, this error vanishes in calculation \eqref{eq: limit}. Thus it follows that the limiting function of $h_\delta(t/\delta^2)$ solves \eqref{eq: 2dimCEAD}. Existence and uniqueness of a solution of equation \eqref{eq: 2dimCEAD} follow from the right hand side being a Lipschitz function on the compact set $I\times J$. As in the proof of Theorem \ref{Thm: CEAD}, since the convergence in \eqref{eq: limit} is uniform in $x_1$ (because the error terms result from the approximation of the derivative), we obtain uniform convergence of the function $\hat{h}_\delta(\cdot/\delta^2)$ which interpolates the intersections of the graph of $\tilde{h}_\delta(\cdot/\delta^2)$ with $\Mcal$ to $x(\cdot)$. Since the $L^\infty$ norm of $\hat{h}_\delta-\tilde{h}_\delta$ tends to $0$ with $\delta\to 0$, we also obtain uniform convergence of $\tilde{h}_\delta(\cdot/\delta^2)$ and hence of $h_\delta(\cdot/\delta^2)$ to $x(\cdot)$.
\end{proof}
\begin{remark}
	Note that the slower speed of the canonical equation is the result of the limited directions of mutation. When the curve given by $\Mcal$ is not in line with the grid (i.e.~not a horizontal or vertical line), we cannot exactly trace $\Mcal$ because the mutations need to remain on the grid. If mutations were allowed in all directions, we would recoved the speed $[\nabla_1f(x(t),x(t))]_1$. 
\end{remark}
\begin{remark}
	We want to discuss shortly the case of a general $n$-dimensional trait space. Assuming again that only mutations to the direct neighbours are permitted and there is always a unique invading trait and all other mutants competing for simultaneous invasion die out, we suspect a similar result to hold. Initially we see one-dimensional behaviour until two coordinates of the gradient $\nabla_1 f$ are equal which w.l.o.g.~are the first two coordinates. Then the solution of the canonical equation will (depending on attractiveness) either be deflected to continue in one-dimensional fashion in the second coordinate, or we will see two-dimensional motion along the set determined by $[\nabla_1f]_1=[\nabla_1f]_2$ until these coordinates are equal to a third coordinate of $\nabla_1f$. This will continue until we eventually have all coordinates of $\nabla_1f$ equal and we see $n$-dimensional motion along this curve. The speed at which we move along these sets is still given by $[\nabla_1f]_1$ divided by a term which describes the length of the shortest path on an $n$-dimensional grid connecting two points on a hyperplane with fixed angles which then is the tangent plane. Since such a hyperplane has $n-1$ angles with respect to the first $n-1$ coordinates of space, the formulas will become increasingly complex.
\end{remark}

	\subsection{Examples}\label{Sec: Examples}
	We want to give examples for the evolution of the resident trait over time in a two-dimensional trait space in an attractive case and an unattractive one. The two coordinates of each trait will be responsible for an increase in reproduction and a decrease in susceptibility to competitive pressure. For visualisation, we give an inequality plot comparing the two components of the gradient $\nabla_1f$ as well as a simulation showing the path that the resident trait takes.  \begin{example}\label{Ex: non attractive}
		Consider the case in which we have $I=J=[0,2]$, $b(x,y)=1+x/2$, $d(x,y)=0$ and $c((x,y),(w,z))=1-y/3$. This choice of rates is made in such a way that the functions are as simple as possible with an interesting effect in the chosen interval. Then, the map of inequalities for the coordinates of $\nabla_1f$ is given in Figure \ref{Fig: ex1} (left). Calculating the gradient, we find that the line at which the coordinates of $\nabla_1f$ are equal is given by $y= 1-x$. This model does not have any evolutionary singular strategies and hence the population will adopt the most beneficial trait in the trait space which is $(2,2)$. We are mostly interested in the path taken from the starting trait $(0,0)$ to the final trait $(2,2)$.
			\begin{figure}[htbp]
			\begin{minipage}{0.48\linewidth}
					\includegraphics[width=\linewidth]{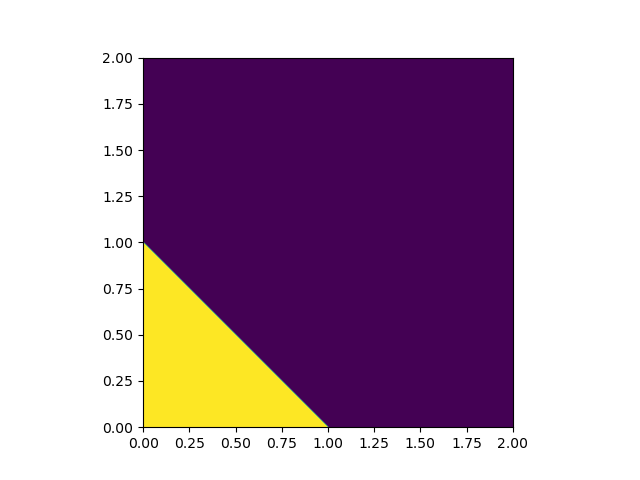}
			\end{minipage}
			\begin{minipage}{0.48\linewidth}
				\includegraphics[width=\linewidth]{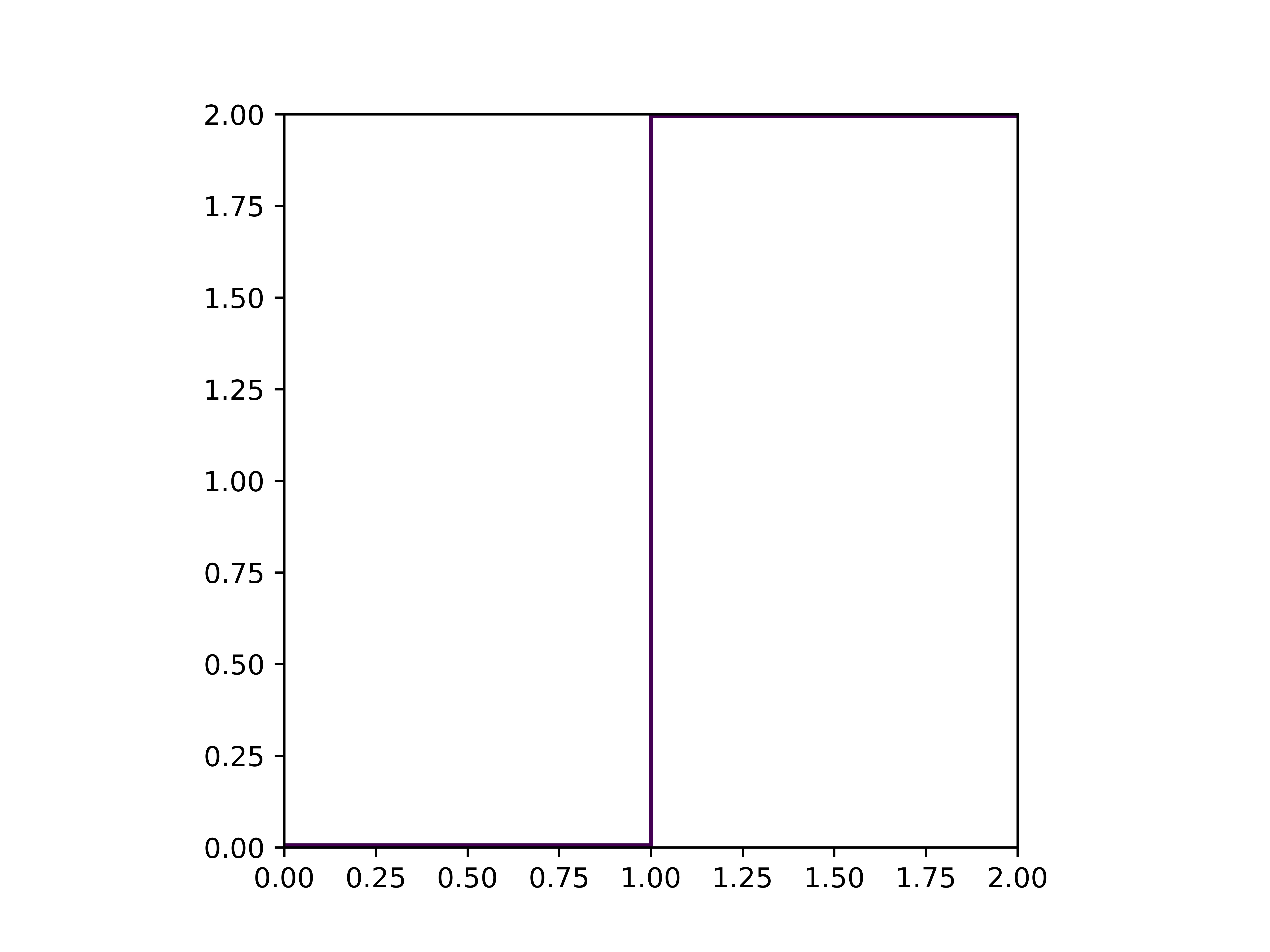}
			\end{minipage}
			\caption{Left: Inequality plot for the parameters in Example \ref{Ex: non attractive}. Both coordinates of the partial derivative are positive. The yellow area satisfies $[\nabla_1f(x,x)]_1>[\nabla_1f(x,x)]_2$ and the black area is $[\nabla_1f(x,x)]_2>[\nabla_1f(x,x)]_1$. In particular, the line where both coordinates are equal is not attractive for $f$. Right: The path in trait space which the sequence of resident traits follows for $\delta=0.011$, started at (0,0).}
			\label{Fig: ex1}
		\end{figure}
	Indeed, we see that the set $\Mcal$ in this case is not attractive and the evolution of the dominating trait takes a turn from going horizontally to going vertically in trait space upon encountering the line $\Mcal$ (Figure \ref{Fig: ex1} (right)). We chose $\delta$ such that we do not have a point exactly on the set $\Mcal$ since this might cause problems with the fitness of both mutant traits being identical (as the derivatives of the fitness are equal).
	\end{example}
	\begin{example}\label{Ex: attractive}
	Now, let $I=J=[0,4]$, $b(x,y)=1+x$, $d(x,y)=0$ and we consider the competition kernel $c((x,y),(w,z))=1+e^{-y}$. Now computing the curve at which the coordinates of $\nabla_1f$ agree, we find that this is the case for $y=\log(x)$. Also, we see from Figure \ref{Fig: ex2} (left), that this curve is attractive for $f$. Again, there is no evolutionary singular strategy and hence the population will evolve to the most beneficial trait, this time being $(4,4)$. The path taken from the initial trait $(0,0)$ in the bottom left corner of the trait space to the final trait $(4,4)$ in the upper right corner is quite different from the one in the previous example.
	\begin{figure}[htbp]
			\begin{minipage}{0.48\linewidth}
			\includegraphics[width=\linewidth]{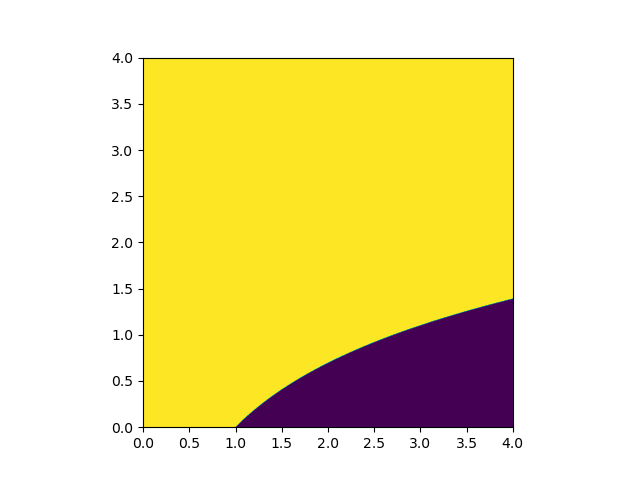}
		\end{minipage}
		\begin{minipage}{0.48\linewidth}
			\includegraphics[width=\linewidth]{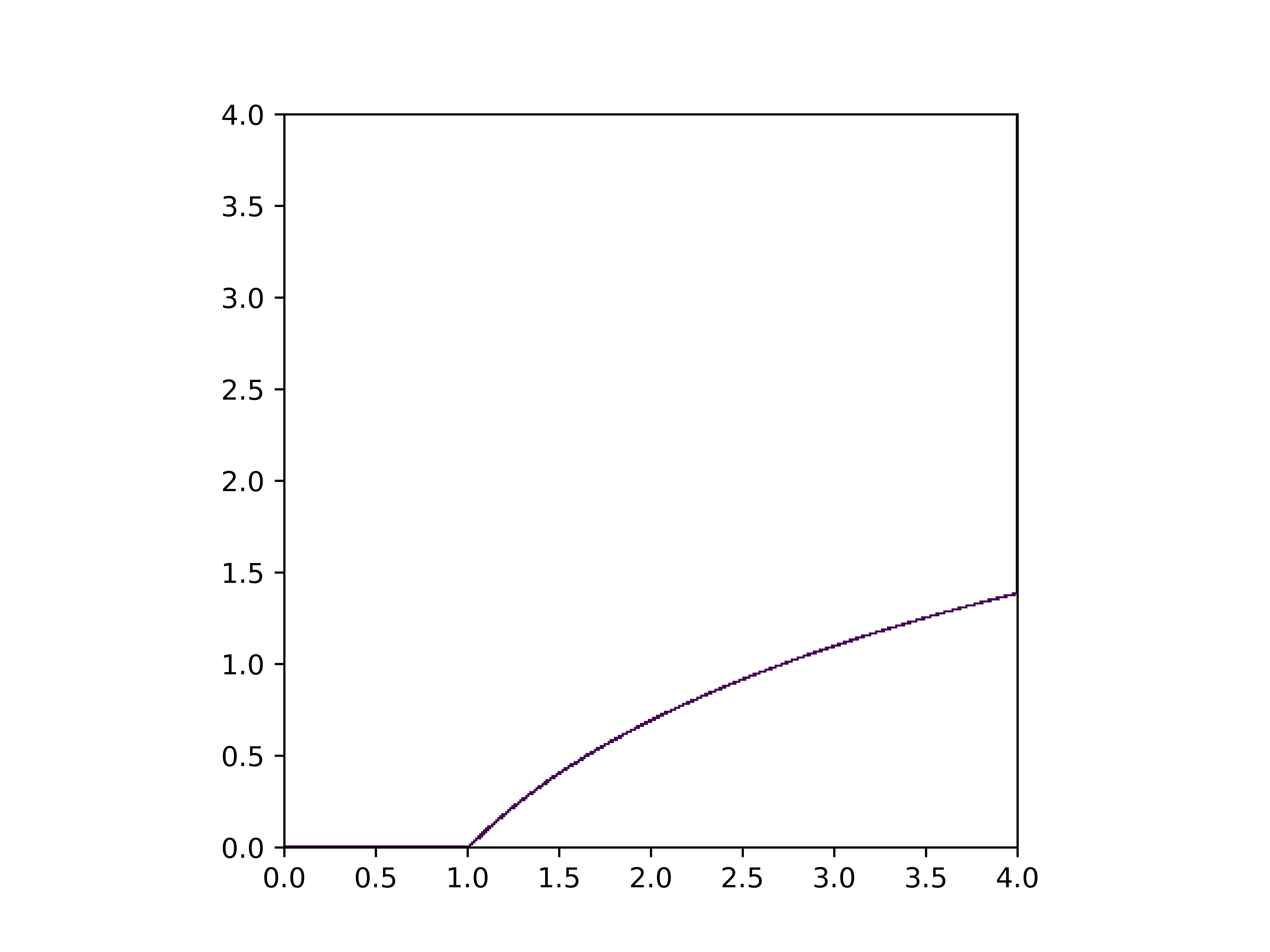}
		\end{minipage}
		\caption{Left: Inequality plot for the parameters in Example \ref{Ex: attractive}. Both coordinates of the partial derivative are positive. The yellow area satisfies $[\nabla_1f(x,x)]_1>[\nabla_1f(x,x)]_2$ and the black area is $[\nabla_1f(x,x)]_2>[\nabla_1f(x,x)]_1$. Here, the set $\Mcal$ is attractive. Right: The path taken by the resident trait for $\delta=0.011$ started at $(0,0)$.}
		\label{Fig: ex2}
	\end{figure}
	\end{example}

 	We plot the Euclidean difference of the solution $x(t)$ to the canonical equation from Theorem \ref{Thm: 2dim CEAD} and the resident trait from our simulation in Figure \ref{Fig: Error}. To improve legibility of this figure, we calculate the difference only at the times when the function $h_\delta$ has jumps and interpolate these errors. 
 	\begin{figure}[htbp]
 		\begin{minipage}{0.45\linewidth}
 		\begin{tikzpicture}
 				\node (img1)  {\includegraphics[width=1.08\linewidth]{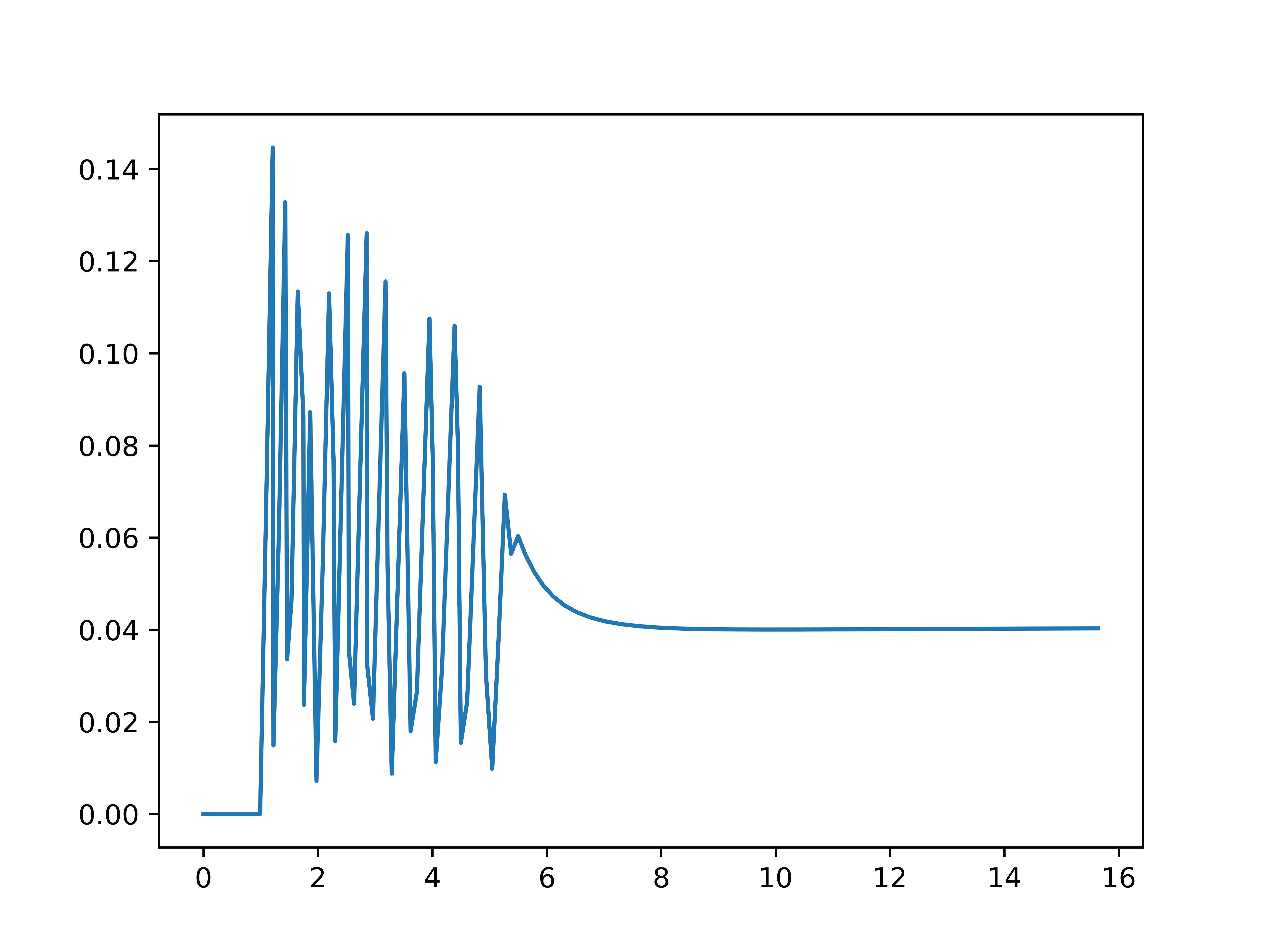}};
 				\node[below=of img1, node distance=0cm, yshift=1.5cm, xshift=0cm] {\footnotesize Time $t$ };
 				\node[left=of img1, node distance=0cm, rotate=90, anchor=center,yshift=-1.3cm] {\footnotesize Error};
 			\end{tikzpicture}
 		\end{minipage}
 		\begin{minipage}{0.45\linewidth}
 			\begin{tikzpicture}
 				\node (img1)  {\includegraphics[width=1.08\linewidth]{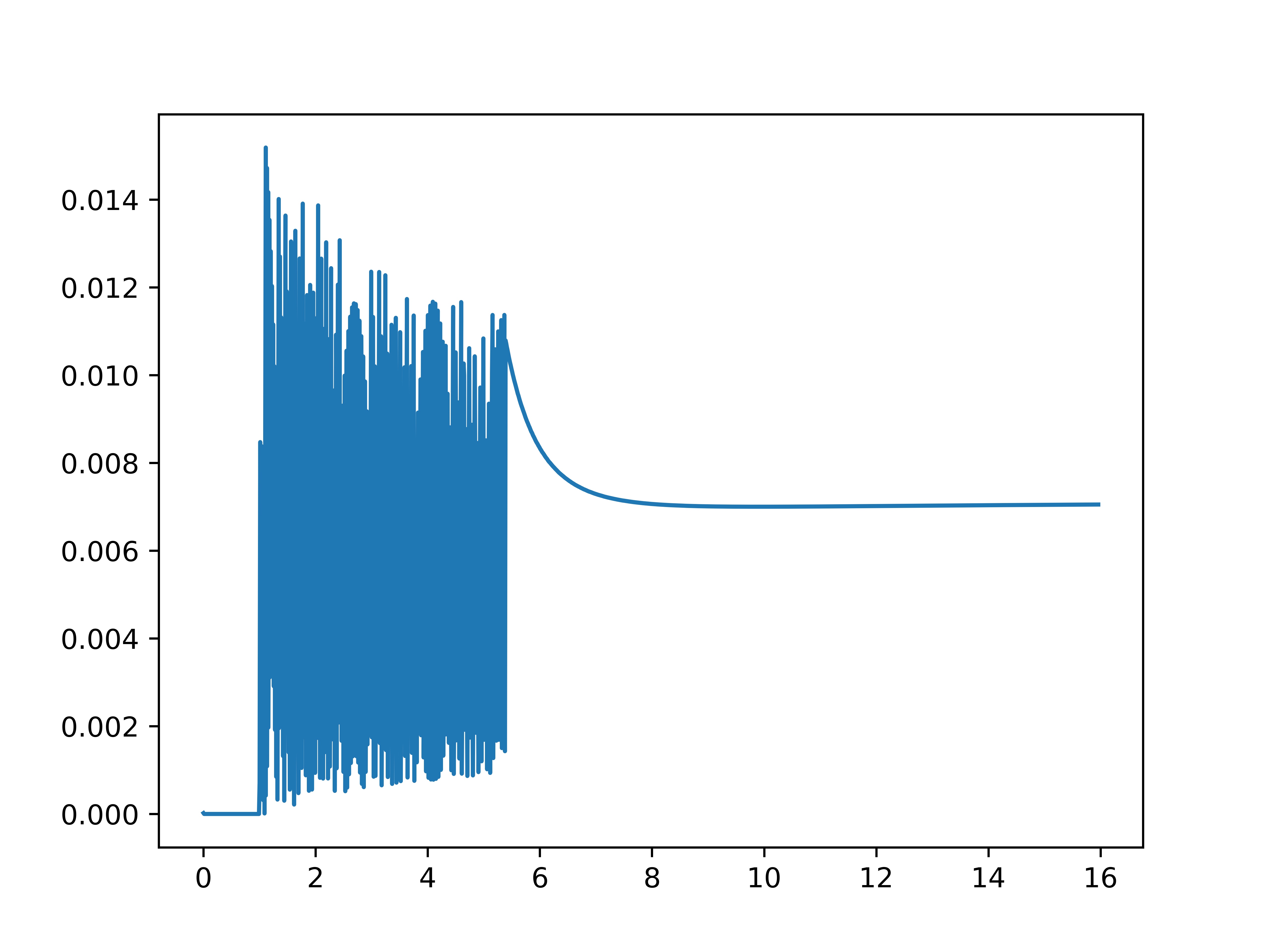}};
 				\node[below=of img1, node distance=0cm, yshift=1.5cm, xshift=-0cm] {\footnotesize Time $t$ };
 				\node[left=of img1, node distance=0cm, rotate=90, anchor=center,yshift=-1.3cm] {\footnotesize Error};
 			\end{tikzpicture}
 		\end{minipage}
 		\caption{Euclidean norm $\lVert x(t)-h_\delta(t/\delta^2)\rVert_2$ for $\delta=0.11$ (left) and $\delta=0.011$ (right). The phase during which $x(t)$ is not a straight line is clearly identifiable.}
 		\label{Fig: Error}
 	\end{figure}
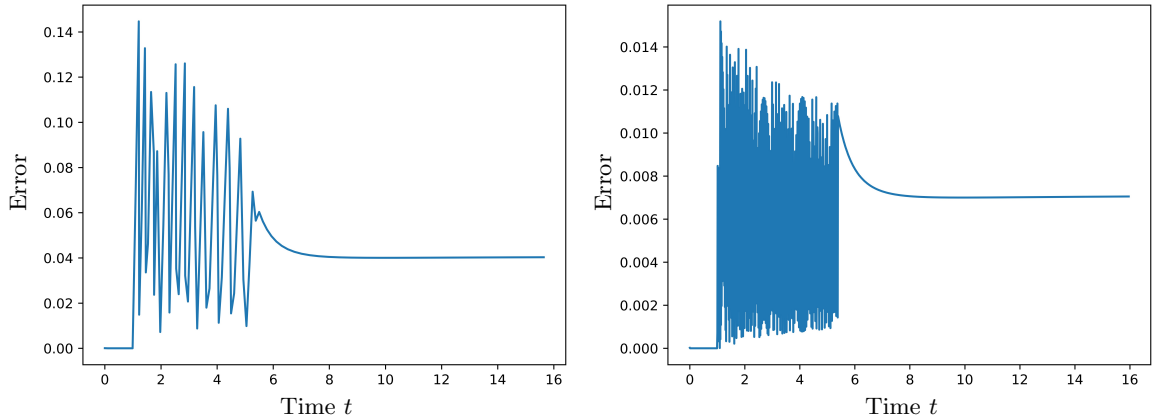
 	As expected, the error is limited by $\sqrt{2}\delta$ which is the furthest distance we should see on a $\delta$ grid from a smooth curve. Interestingly however, we see that at first we have a negligible error until we encounter the critical curve $\Mcal$. Here, the error is erratic since jumps crossing the curve may increase the distance from the curve. As soon as we reach the boundary of our trait space though, the error is smooth again but remains at a non-negligible level due to the $\delta$-grid not being exactly on the edge of the trait space. When both processes reach the upper right corner of the trait space, the error remains at a constant level.
	\appendix
	\section{Technical details for the simulations}\label{Sec:Simulation details}
	Here, we give a short summary of the techniques used to create our examples.
	\subsection{Simulations in one dimension}
	
	The solution $x$ of the canonical equation in Figure \ref{Fig: Comparison 1dim CEAD} was computed by using a simple Euler method and computing the partial derivative of the fitness function $f$ by setting $\partial_1 f(x,x)\approx f(x+h,x)/h$ for $h=10^{-5}$. For the function $g_\delta$, we used the explicit form detailed before Theorem \ref{Thm: CEAD}.\\
	
	For Figure \ref{Fig: 1dim branching} we require the computation of the coexistence equilibria. In accordance with Assumption \ref{asmpt: stability}, we used a damped Newton method with damping factor $0.02$ and initial condition composed of the old equilibrium population for previously resident traits and a small starting population of $0.01$ for invading mutant traits. Computations for the equilibrium were done for at most $10000$ iterations.
	
	\subsection{Simulations in two dimensions}
	
	The inequality plots for Figures \ref{Fig: ex1} and \ref{Fig: ex2} were computed on a $1000\times 1000$ grid on the square $[0,2]^2$ and $[0,4]^2$ respectively by calculating the approximate gradient again using $[\nabla_1f(x,x)]_i\approx f(x+h\mathbf{e}_i,x)/h$ with $h=10^{-5}$ for each of the grid points. If the difference of the components was less than $10^{-5}$, they were set to be equal.\\
	
	The plots showing the path of the resident trait in trait space again required the calculation of various equilibria. Thse were done as for the evolutionary branching images.\\
	
	To calculate the error in Figure \ref{Fig: Error}, we first explicitly computed the speed at which we move along the curve $\Mcal=\{(t,\log(t))\mid t\in[1,4]\}$. This is given in Theorem \ref{Thm: 2dim CEAD} as \[
	\frac{[\nabla_1f(x(t),x(t))]_1}{\sin(\angle x(t))+\cos(\angle x(t))},
	\]
	For $x(t)\in\Mcal$ it is easy to compute $[\nabla_1f(x(t),x(t))]_1=1$. The tangent at the point $x(t)=(a,\log(a))$ is the line uniquely determined by \[
	x(t)+y\cdot\begin{bmatrix}
		1\\ 1/a
	\end{bmatrix},\qquad y\in\R.
	\]
	In particular, the angle with the positive real line is $\arctan(1/a)$. Plugging this into the formula for the speed at which we move along $\Mcal$ and normalizing the tangent vector, we obtain \[
	x'(t)=\frac{[x(t)]_1}{1+[x(t)]_1}\cdot\begin{bmatrix}
		1\\ 1/[x(t)]_1
	\end{bmatrix}
	\]
	Then we separated the three phases and solved that canonical equation piecewise using an Euler scheme with $\Delta t=T\cdot 10^{-6}$ on a time horizon of $T=16$. Having simulated the times at which the resident trait changes, we compared the discrepancy of the simulated resident trait at these times (accelerated by $1/\delta^2$) and the solution of the canonical equation.

	\section*{Acknowledgements}
	The author wants to thank M. Wilke Berenguer for helpful comments on an earlier version of this paper. This work was funded by the Deutsche Forschungsgemeinschaft (DFG, German Research Foundation) under Germany’s Excellence Strategy MATH+: The Berlin Mathematics Research Center, EXC-2046/1 project-ID 390685689.

\end{document}